\pdfoutput=1
\documentclass{lmcs} 
\pdfoutput=1

\usepackage{lastpage}
\lmcsdoi{21}{1}{30}
\lmcsheading{}{\pageref{LastPage}}{}{}%
{Apr.~18,~2023}{Mar.~28,~2025}{}

\usepackage[utf8]{inputenc}
\usepackage[T1]{fontenc}
\usepackage{etex}

\usepackage{epsfig}
\usepackage{latexsym}
\usepackage{graphicx}
\usepackage{amssymb}
\usepackage{amsfonts}
\usepackage{amsthm}
\usepackage{enumitem}
\usepackage{amsmath}
\usepackage{verbatim}
\usepackage{url}
\usepackage{colortbl}
\usepackage{stmaryrd}
\usepackage{epic,eepic}
\usepackage{xspace}
\usepackage{tikz}
\usepackage[noend]{algpseudocode}
\usepackage{tcolorbox}
\usepackage{multirow}
\usetikzlibrary{automata}
\usetikzlibrary{arrows}
\usetikzlibrary{backgrounds}
\usetikzlibrary{positioning}
\usetikzlibrary{fit}
\usetikzlibrary{calc}
\usetikzlibrary{matrix}
\usetikzlibrary{patterns}
\usepackage{stackengine}
\usepackage{xspace}
\usepackage{xcolor}
\usepackage[linesnumbered,ruled,vlined]{algorithm2e}

\SetCommentSty{mycommfont}
\SetKwInput{KwData}{Input}

\usepackage{microtype}
\usepackage{ellipsis}         % improves whitespacing around "..."s
\usepackage[abbreviations,british]{foreign} % for better typing i.e., e.g. etc. 
\usetikzlibrary{shadows}
\tikzset{
diagonal fill/.style 2 args={fill=#2, path picture={
\fill[#1, sharp corners] (path picture bounding box.south west) -|
                         (path picture bounding box.north east) -- cycle;}},
reversed diagonal fill/.style 2 args={fill=#2, path picture={
\fill[#1, sharp corners] (path picture bounding box.north west) |- 
                         (path picture bounding box.south east) -- cycle;}}
}

 \makeatletter
 \def\desclabel#1#2{\begingroup
    \def\@currentlabel{#1}%
    #1\label{#2}\endgroup
 }
 \makeatother

\usetikzlibrary{patterns}
\usetikzlibrary{fit}
\usetikzlibrary{calc}
\usetikzlibrary{decorations.pathreplacing}
\usetikzlibrary{decorations.text}

\usepackage{xparse}
\usepackage{mathdots}

\usepackage{float}
\usepackage{thmtools}
\usepackage{thm-restate}
\usepackage{pdfpages}
\usepackage{hyperref}
\usepackage{fontawesome5}
\usepackage{scalerel}

\definecolor{mygreen}{RGB}{36, 200, 100}
\definecolor{myyellow}{RGB}{220, 220, 30}

\definecolor{ao(english)}{rgb}{0.0, 0.5, 0.0}

\newcommand{\mathcommand}[2]{\newcommand{#1}{\ensuremath{#2}\xspace}}

\newcommand{\mylogic}[1]{\textsc{#1}}

\mathcommand{\logic}{\mathcal L}
\mathcommand{\logicbis}{\logic'}

\mathcommand{\FO}{\mylogic{FO}}
\mathcommand{\FOtwo}{\FO^2}
\mathcommand{\FOt}{\FOtwo}
\mathcommand{\FOthree}{\FO^3}
\mathcommand{\Ctwo}{\mylogic{C}^2}
\mathcommand{\LFP}{\mylogic{LFP}}
\mathcommand{\PFP}{\mylogic{PFP}}
\mathcommand{\MSO}{\mylogic{MSO}}
\mathcommand{\ESO}{\mylogic{ESO}}
\mathcommand{\USO}{\mylogic{USO}}
\mathcommand{\CMSO}{\mylogic{CMSO}}

\mathcommand{\oi}{<\!\textup{-inv }}
\mathcommand{\oilogic}{\oi\logic}
\mathcommand{\oifo}{\oi\FO}
\mathcommand{\oifotwo}{\oi\FOtwo}
\mathcommand{\oifothree}{\oi\FOthree}
\mathcommand{\oictwo}{\oi\Ctwo}
\mathcommand{\oimso}{\oi\MSO}
\mathcommand{\oilfp}{\oi\LFP}
\mathcommand{\oipfp}{\oi\PFP}

\mathcommand{\si}{\textup{Succ-inv }}
\mathcommand{\sifo}{\si\FO}
\mathcommand{\simso}{\si\MSO}

\mathcommand{\arbi}{\textup{Arb-inv }}
\mathcommand{\arbifo}{\arbi\FO}

\NewDocumentCommand{\logeq}{O{k}}{\equiv^{{\logic}}_{#1}}
\NewDocumentCommand{\foeq}{O{k}}{\equiv^{{\FO}}_{#1}}
\NewDocumentCommand{\foneq}{O{k}}{\not\equiv^{{\FO}}_{#1}}
\NewDocumentCommand{\fotwoeq}{O{k}}{\equiv^{{\FOtwo}}_{#1}}
\NewDocumentCommand{\ctwoeq}{O{k,k}}{\equiv^{{\Ctwo}}_{#1}}
\NewDocumentCommand{\msoeq}{O{k}}{\equiv^{{\MSO}}_{#1}}
\NewDocumentCommand{\oifoeq}{O{k}}{\equiv^{{\oifo}}_{#1}}
\NewDocumentCommand{\oifotwoeq}{O{k}}{\equiv^{{\oifotwo}}_{#1}}
\NewDocumentCommand{\oictwoeq}{O{k,k}}{\equiv^{{\oictwo}}_{#1}}
\NewDocumentCommand{\sifoeq}{O{k}}{\equiv^{{\sifo}}_{#1}}

\newcommand{\myclass}[1]{\textsc{#1}}

\mathcommand{\PTIME}{\myclass{PTime}}
\mathcommand{\LOGSPACE}{\myclass{LogSpace}}
\mathcommand{\NLOGSPACE}{\myclass{NLogSpace}}
\mathcommand{\PSPACE}{\myclass{PSpace}}
\mathcommand{\ACO}{\myclass{AC}^0}
\mathcommand{\NCO}{\myclass{NC}^0}
\mathcommand{\TCO}{\myclass{TC}^0}
\mathcommand{\coNET}{\myclass{coNExpTime}}
\mathcommand{\NET}{\myclass{NExpTime}}
\mathcommand{\PT}{\myclass{PTime}}
\mathcommand{\coNtwoET}{\myclass{coN2ExpTime}}
\mathcommand{\Ackermann}{\myclass{Ackermann}}

\newcommand{\nocc}[2]{|#1|_{#2}}

\newcommand{\numtp}[1]{\llbracket#1\rrbracket}
\NewDocumentCommand{\boule}{O{\elem}O{\struct}O{k}}{N_{#2}^{#3}(#1)}
\NewDocumentCommand{\boulegen}{O{\elem}O{\structgen}O{k}}{N_{#2}^{#3}(#1)}
\NewDocumentCommand{\neigh}{O{\elem}O{\struct}O{k}}{\mathcal N_{#2}^{#3}(#1)}
\NewDocumentCommand{\neighgen}{O{\elem}O{\structgen}O{k}}{\mathcal N_{#2}^{#3}(#1)}

\newcommand{\threq}[4]{\numtp{#1}_{#3}=^{#4}\numtp{#2}_{#3}}
\NewDocumentCommand{\voisdom}{O{\pebbletwo}O{\struct}}{\boule[#1][#2][1]}

\NewDocumentCommand{\vois}{O{\pebbletwo}O{\struct}O{r}}{\neigh[#1][#2][#3]}

\NewDocumentCommand{\ntpgen}{O{\elem}O{\structgen}O{k}}{\text{neigh-tp}_{#2}^{#3}(#1)}
\NewDocumentCommand{\ntp}{O{\elem}O{k}}{\ntpgen[#1][\struct][#2]}
\NewDocumentCommand{\ntpbis}{O{\elem}O{k}}{\ntpgen[#1][\structbis][#2]}
\NewDocumentCommand{\ntpind}{O{\elem}O{k}}{\ntpgen[#1][\structind][#2]}

\NewDocumentCommand{\tpgen}{O{\elem}O{\structgen}}{\text{tp}_{#2}(#1)}
\NewDocumentCommand{\tp}{O{\pebbleone,\pebbletwo}}{\tpgen[#1][\orderstruct]}
\NewDocumentCommand{\tpbis}{O{\pebbleonebis,\pebbletwobis}}{\tpgen[#1][\orderstructbis]}

\NewDocumentCommand{\ordtpgen}{O{\elem,\elembis}O{\ordgen}}{\text{tp}_{#2}(#1)}
\NewDocumentCommand{\ordtp}{O{\pebbleone,\pebbletwo}}{\ordtpgen[#1][\ordgen]}
\NewDocumentCommand{\ordtpbis}{O{\pebbleonebis,\pebbletwobis}}{\ordtpgen[#1][\ordgen]}

\NewDocumentCommand{\vocabtp}{O{\pebbleone,\pebbletwo}}{\tpgen[#1][\struct]}
\NewDocumentCommand{\vocabtpbis}{O{\pebbleonebis,\pebbletwobis}}{\tpgen[#1][\structbis]}

\newcommand{\set}{P_0}
\newcommand{\setbis}{P_1}

\newcommand{\vocabtpone}{\vocabtp[\pebbleone]}
\newcommand{\vocabtptwo}{\vocabtp[\pebbletwo]}
\newcommand{\vocabtponebis}{\vocabtpbis[\pebbleonebis]}
\newcommand{\vocabtptwobis}{\vocabtpbis[\pebbletwobis]}

\mathcommand{\N}{\mathbb{N}}
\mathcommand{\im}{\operatorname{Im}}

\newcommand{\dist}[3]{\text{dist}_{#1}(#2,#3)}

\newcommand{\EF}{Ehrenfeucht-Fra\"iss\'e\xspace}

\newcommand{\existsm}[1]{\exists^{\geq #1}}

\NewDocumentCommand{\inter}{O{}O{}}{\mathcal I_{#1}^{#2}}

\mathcommand{\elem}{a}
\mathcommand{\elembis}{b}
\mathcommand{\elemter}{c}

\mathcommand{\var}{x}
\mathcommand{\varbis}{y}

\mathcommand{\classe}{\mathdutchcal{C}}
\mathcommand{\classebis}{\classe'}
\mathcommand{\propriete}{\mathdutchcal{P}}

\mathcommand{\ind}{i}

\mathcommand{\indin}{\ind\in\{0,1\}}

\mathcommand{\structdom}{A_0}
\mathcommand{\struct}{\mathcal A_0}
\mathcommand{\structbisdom}{A_1}
\mathcommand{\structbis}{\mathcal A_1}
\mathcommand{\structinddom}{A_\ind}
\mathcommand{\structind}{\mathcal A_\ind}
\mathcommand{\structgendom}{A}
\mathcommand{\structgen}{\mathcal A}

\mathcommand{\ord}{<_0}
\mathcommand{\ordbis}{<_1}
\mathcommand{\ordind}{<_\ind}
\mathcommand{\ordgen}{<}

\mathcommand{\orderstruct}{(\struct,\ordgen)}
\mathcommand{\orderstructbis}{(\structbis,\ordgen)}
\mathcommand{\orderstructind}{(\structind,\ordgen)}
\mathcommand{\orderstructgen}{(\structgen,\ordgen)}

\mathcommand{\type}{\tau}
\NewDocumentCommand{\otype}{O{l}}{\type_{#1}}
\NewDocumentCommand{\otypefirst}{O{1}}{\otype[#1]}
\NewDocumentCommand{\otypelast}{O{|\Ord|}}{\otype[#1]}

\mathcommand{\formule}{\varphi}
\mathcommand{\formulebis}{\psi}
\mathcommand{\formuleter}{\theta}
\mathcommand{\vocab}{\Sigma}
\mathcommand{\vocabbis}{\vocab'}
\mathcommand{\arbre}{\mathcal T}

\mathcommand{\graphe}{\mathcal\graphedom}
\mathcommand{\graphedom}{G}
\mathcommand{\graphedombis}{H}
\mathcommand{\edgerel}{E}

\mathcommand{\gaifman}{\graphe_{\structgen}}
\NewDocumentCommand{\powgaifman}{O{\delta}}{\gaifman^{#1}}

\NewDocumentCommand{\vocabgraphe}{sO{\vocab}O{\edgerel}}{
  \IfBooleanTF#1
              {#2_{#3}=\{#3^{(2)}\}}
              {#2_{#3}}
} 

\NewDocumentCommand{\structgraphe}{sO{\graphedom}O{\edgerel}}{
  \IfBooleanTF#1
              {\mathcal{#2}=(#2,#3^{\mathcal{#2}})}
              {\mathcal{#2}}
}
\NewDocumentCommand{\structgraphebis}{sO{\graphedombis}O{\edgerel}}{
  \IfBooleanTF#1
              {\mathcal{#2}=(#2,#3^{\mathcal{#2}})}
              {\mathcal{#2}}
} 

\mathcommand{\arbredom}{T}
\mathcommand{\childrel}{\edgerel}

\NewDocumentCommand{\structarbre}{sO{\arbredom}O{\childrel}}{
  \IfBooleanTF#1
              {\structgraphe*[#2][#3]}
              {\structgraphe[#2][#3]}
} 

\mathcommand{\structarbrebis}{\structarbre'}

\mathcommand{\treedec}{(\structarbre,\bag)}

\mathcommand{\relsymb}{R}

\NewDocumentCommand{\pebblegen}{O{i}O{\alpha}}{p_{#1}^{#2}}
\NewDocumentCommand{\pebbleone}{}{\pebblegen[0][x]}
\NewDocumentCommand{\pebbletwo}{}{\pebblegen[0][y]}
\NewDocumentCommand{\pebbleonebis}{}{\pebblegen[1][x]}
\NewDocumentCommand{\pebbletwobis}{}{\pebblegen[1][y]}
\newcommand{\epsin}{\alpha\in\{x,y\}}

\NewDocumentCommand{\bijsymb}{O{0}}{\phi_{#1}}
\NewDocumentCommand{\bijgen}{O{\elem}O{i}}{\bijsymb[#2](#1)}
\NewDocumentCommand{\bij}{O{\pebbleone}}{\bijgen[#1][0]}
\NewDocumentCommand{\bijbis}{O{\pebbleonebis}}{\bijgen[#1][1]}

\NewDocumentCommand{\Ord}{O{\type}}{\ensuremath{\textsc{Env}(#1)}\xspace}
\NewDocumentCommand{\ordORD}{O{k}O{d}O{\vocab}}{\prec_{#1}^{#2}}
\NewDocumentCommand{\Ntp}{O{k}O{d}O{\vocab}}{\textsc{NeighType}_{#1}^{#2}}
\NewDocumentCommand{\Freq}{O{\struct}O{k}O{}}{\textsc{Freq}[#1]_{#2}^{#3}}
\NewDocumentCommand{\nbrtp}{O{k}O{d}}{\ensuremath{\#\text{\small neigh}_{#1}^{#2}}\xspace}

\NewDocumentCommand{\pin}{O{\otype}O{*}O{j}}{\elem[#1]_{#2}^{#3}}
\NewDocumentCommand{\pinLeft}{O{\otype}O{j}}{\pin[#1][L][#2]}
\NewDocumentCommand{\pinRight}{O{\otype}O{j}}{\pin[#1][R][#2]}

\NewDocumentCommand{\envtpgen}{O{\elem}O{\orderstructgen}O{k}}{\text{env-tp}_{#2}^{#3}(#1)}
\NewDocumentCommand{\envtp}{O{\elem}O{k}}{\envtpgen[#1][\orderstruct][#2]}
\NewDocumentCommand{\envtpbis}{O{\elem}O{k}}{\envtpgen[#1][\orderstructbis][#2]}

\NewDocumentCommand{\Envgen}{O{\elem}O{\orderstructgen}O{k}}{\text{Env}_{#2}^{#3}(#1)}
\NewDocumentCommand{\Env}{O{\elem}O{k}}{\Envgen[#1][\orderstruct][#2]}
\NewDocumentCommand{\Envbis}{O{\elem}O{k}}{\Envgen[#1][\orderstructbis][#2]}

\NewDocumentCommand{\Leftgen}{O{}O{j}}{L_{#1}^{#2}}
\NewDocumentCommand{\Left}{O{j}}{\Leftgen[0][#1]}
\NewDocumentCommand{\Leftbis}{O{j}}{\Leftgen[1][#1]}
\NewDocumentCommand{\Leftind}{O{j}}{\Leftgen[\ind][#1]}

\NewDocumentCommand{\NLeftgen}{O{}O{j}}{N\!L_{#1}^{#2}}
\NewDocumentCommand{\NLeft}{O{j}}{\NLeftgen[0][#1]}
\NewDocumentCommand{\NLeftbis}{O{j}}{\NLeftgen[1][#1]}
\NewDocumentCommand{\NLeftind}{O{j}}{\NLeftgen[\ind][#1]}

\NewDocumentCommand{\ULeftgen}{O{}O{j}}{U\!L_{#1}^{#2}}
\NewDocumentCommand{\ULeft}{O{j}}{\ULeftgen[0][#1]}
\NewDocumentCommand{\ULeftbis}{O{j}}{\ULeftgen[1][#1]}
\NewDocumentCommand{\ULeftind}{O{j}}{\ULeftgen[\ind][#1]}

\NewDocumentCommand{\Rightgen}{O{}O{j}}{R_{#1}^{#2}}
\NewDocumentCommand{\Right}{O{j}}{\Rightgen[0][#1]}
\NewDocumentCommand{\Rightbis}{O{j}}{\Rightgen[1][#1]}
\NewDocumentCommand{\Rightind}{O{j}}{\Rightgen[\ind][#1]}

\NewDocumentCommand{\NRightgen}{O{}O{j}}{N\!R_{#1}^{#2}}
\NewDocumentCommand{\NRight}{O{j}}{\NRightgen[0][#1]}
\NewDocumentCommand{\NRightbis}{O{j}}{\NRightgen[1][#1]}
\NewDocumentCommand{\NRightind}{O{j}}{\NRightgen[\ind][#1]}

\NewDocumentCommand{\URightgen}{O{}O{j}}{U\!R_{#1}^{#2}}
\NewDocumentCommand{\URight}{O{j}}{\URightgen[0][#1]}
\NewDocumentCommand{\URightbis}{O{j}}{\URightgen[1][#1]}
\NewDocumentCommand{\URightind}{O{j}}{\URightgen[\ind][#1]}

\NewDocumentCommand{\Raregen}{O{}}{X_{#1}}
\mathcommand{\Rare}{\Raregen[0]}
\mathcommand{\Rarebis}{\Raregen[1]}
\mathcommand{\Rareind}{\Raregen[\ind]}

\NewDocumentCommand{\Middlegen}{O{}}{M_{#1}}
\mathcommand{\Middle}{\Middlegen[0]}
\mathcommand{\Middlebis}{\Middlegen[1]}
\mathcommand{\Middleind}{\Middlegen[\ind]}

\NewDocumentCommand{\Segmentgen}{O{}O{j}}{S_{#1}^{#2}}
\NewDocumentCommand{\Segment}{O{r}}{\Segmentgen[0][#1]}
\NewDocumentCommand{\Segmentbis}{O{r}}{\Segmentgen[1][#1]}
\NewDocumentCommand{\Segmentind}{O{r}}{\Segmentgen[\ind][#1]}

\NewDocumentCommand{\Neighgen}{mO{\structgen}}{N_{#2}(#1)}
\NewDocumentCommand{\Neigh}{m}{\Neighgen{#1}[\struct]}

\NewDocumentCommand{\nextordgen}{O{\ordgen}}{\cdot_{#1}}

\mathcommand{\thr}{\Upsilon}

\NewDocumentCommand{\labsr}{O{r}}{\mathcal{S}_{#1}}
\NewDocumentCommand{\refsr}{O{r}}{$(\labsr[#1])$\xspace}

\NewDocumentCommand{\laber}{O{r}}{\mathcal{E}_{#1}}
\NewDocumentCommand{\refer}{O{r}}{$(\laber[#1])$\xspace}

\NewDocumentCommand{\labtr}{O{r}}{\mathcal{T}_{#1}}
\NewDocumentCommand{\reftr}{O{r}}{$(\labtr[#1])$\xspace}

\DeclareMathAlphabet{\mathdutchcal}{U}{dutchcal}{m}{n}
\SetMathAlphabet{\mathdutchcal}{bold}{U}{dutchcal}{b}{n}
\DeclareMathAlphabet{\mathdutchbcal}{U}{dutchcal}{b}{n}

\renewcommand{\restriction}{\mathord{\upharpoonright}}
\newcommand{\restr}[2]{#1\restriction_{#2}} % the restriction of #1 to #2

\usepackage{mathrsfs}
\newcommand{\AAA}{\mbox{\large \boldmath $\alpha$}}
\newcommand{\BBB}{\mbox{\large \boldmath $\beta$}}

\title[Order-Invariance with Two Variables]{About the Expressive Power and Complexity of Order-Invariance with Two Variables}%

\author[B.~Bednarczyk]{Bartosz Bednarczyk\lmcsorcid{0000-0002-8267-7554}}[a]
\author[J.~Grange]{Julien Grange\lmcsorcid{0009-0005-0470-1781}}[b]
\address{Institute of Computer Science, University of Wrocław}
\email{bartosz.bednarczyk@cs.uni.wroc.pl}
\address{LACL, Université Paris-Est Créteil, France}
\email{julien.grange@lacl.fr}


\begin{document}

\maketitle

%!TEX root = ../order-invariant-FO2.tex

\begin{abstract}
  Order-invariant first-order logic is an extension of first-order logic ($\FO$)
  where formulae can make use of a linear order on the structures, 
  under the proviso that they are order-invariant, 
  \ie that their truth value is the same for all linear orders.
  We continue the study of the two-variable fragment of order-invariant 
  first-order logic initiated by Zeume and Harwath, and study its complexity and expressive power.
  We first establish $\coNET$-completeness for the problem of deciding if a given two-variable 
  formula is order-invariant, which tightens and significantly simplifies the $\coNtwoET$ proof by Zeume and Harwath.
  Second, we address the question of whether every property expressible in order-invariant two-variable logic is also expressible in first-order logic without the use of a linear order.
  %While we were not able to provide a satisfactory answer to the question, we suspect that the answer is ``no''. 
  We suspect that the answer is ``no''. To justify our claim, we present
  a class of finite tree-like structures (of unbounded degree) in which a relaxed variant of order-invariant two-variable $\FO$ expresses properties that are not definable in plain $\FO$.
  By contrast, we show that if one restricts their attention to classes of structures of bounded degree, then the expressive power of order-invariant two-variable $\FO$ is contained within~$\FO$.
\end{abstract}

%!TEX root = ../order-invariant-FO2.tex

\section{Introduction}\label{sec:introduction}
The main goal of finite model theory is to understand formal languages describing finite structures: their complexity and their expressive power.
Such languages are ubiquitous in computer science, starting from descriptive complexity, where they are used to provide machine-independent characterisations of complexity classes, and ending up on database theory and knowledge representation, where formal languages serve as fundamental querying formalism.
A classical idea in finite model theory is to employ invariantly-used relations, capturing the data-independence principle in databases: it makes sense to give queries the ability to exploit the presence of the order in which the data is stored in the memory, but at the same time we would like to make query results independent of this specific ordering.
It is not immediately clear that the addition of an invariantly-used linear order to first-order logic ($\FO$) allow us to gain anything from the standpoint of expressive power.
Indeed, as long as we consider arbitrary (\ie not necessarily finite) structures it does not, which is a direct consequence of $\FO$ having the Craig Interpolation Property.
However, as it was first shown by Gurevich~\cite[Thm.~5.3]{Libkin04}, the claim holds true over finite structures: order-invariant $\FO$ is more expressive~than~plain~$\FO$ (already the four-variable fragment of $\FO$ can express properties that are not $\FO$-definable).

Unfortunately, order-invariant $\FO$ is poorly understood. 
As stated in~\cite{BarceloL16}, one of the reasons why progress in understanding order-invariance is rather slow is the lack of logical toolkit.
The classical model-theoretic methods based on types were proposed only recently~\cite{BarceloL16}, and order-invariant $\FO$ is not even a logic in the classical sense, as its syntax is undecidable. 
Moreover, the availability of locality-based methods is limited: order-invariant $\FO$ is known to be Gaifman-local~\cite[Thm. 2]{GroheS00} but the status of its Hanf-locality remains open.
This suggests that a good way to understand order-invariant $\FO$ is to first look at its fragments, \eg the fragments with a limited~number~of~variables.

\paragraph*{Our contribution}\label{subsec:our-contribution}
We continue the line of research initiated in~\cite{ZeumeH16}, which aims to study the complexity and the expressive power of order-invariant $\FOt$, the two-variable fragment of order-invariant $\FO$.
From a complexity point of view, it is known that order-invariant $\FOt$ has a $\coNET$-complete 
validity problem (which is inherited from $\FOt$ with a single linear order, see~\cite[Thm. 1.2]{Otto01}),
and that whether a given $\FOt$-formula is order-invariant is decidable in~$\coNtwoET$~\cite[Thm.~12]{ZeumeH16}.
From an expressive power point of view, order-invariant $\FOt$ is more expressive than plain $\FOt$ 
as it can count globally, see~\cite[Example 2]{ZeumeH16}. 
It remains open~\cite[Sec. 7]{ZeumeH16}, however, whether it is true that every order-invariant $\FOt$-formula 
is equivalent to an $\FO$-formula without the linear order predicate. 
This paper contributes to the field in the three following ways: 

\begin{itemize}\itemsep0em
\item We provide a tight bound for 
deciding order-invariance for $\FOt$; namely, we show that this problem is $\coNET$-complete. 
Our proof method relies on establishing an exponential-size counter-model property for order-invariance, 
and is significantly easier than the proof of Zeume and Harwath~\cite[Thm.~12]{ZeumeH16}.

\item We present a class $\classe_{\textit{tree}}$ of tree-like structures, 
inspired by~\cite{Potthoff94}, and show that there exists an $\FOt$-formula that 
is \emph{order-invariant over} $\classe_{\textit{tree}}$ (but not over all finite structures!)
which is not equivalent to any $\FO$-formula without the linear order predicate.
This leads us to believe that the answer to the question of~\cite[Sec. 7]{ZeumeH16} of whether the expressive power of order-invariant $\FOt$ lies inside $\FO$ is ``\emph{no}''. The problem remains open, though.

\item In stark contrast to the previous result, we show that order-invariant $\FOt$ 
cannot express properties beyond the scope of $\FO$ over classes of structures of bounded degree.\@ 
We show that this upper bound remains when adding counting to $\FOtwo$. 
\end{itemize}

This work is an extended version of~\cite{Bednarczyk22} and~\cite{Grange23}. 

%!TEX root = ../order-invariant-FO2.tex

\section{Preliminaries}\label{sec:prelim}

\noindent We employ standard terminology from finite model theory, assuming that the reader is familiar with the syntax and the semantics of first-order logic ($\FO$)~\cite[Sec.~2.1]{Libkin04}, basics on computability and complexity~\cite[Secs.~2.2--2.3]{Libkin04}, and order-invariant queries~\cite[Secs.~5.1--5.2]{Libkin04}.
By $\FO(\Sigma)$ we denote the first-order logic with equality 
(written \FO when \vocab is clear from the context) on a finite signature \vocab composed of 
relation and constant symbols. We say that \vocab is \emph{purely relational} if it contains no constant symbol.
By~\FOtwo we denote the fragment of \FO in which only two variables \var and \varbis are used.

As~usual in mathematics, for a given set $S$ we write $|S|$ to denote its cardinality.

\paragraph*{Structures}\label{para:intro-structures}
Structures are denoted by calligraphic upper-case letters $\mathcal{A}, \mathcal{B}$ and their domains are denoted by the corresponding Roman letters $A, B$.
We assume that structures have non-empty, \emph{finite} domains.
We write $\varphi[{R}{/}{S}]$ to denote the formula obtained from $\varphi$ by replacing each occurrence of the symbol $R$ with~$S$.
We write~$\varphi(\bar{x})$ to indicate that all the free variables of~$\varphi$ are in~$\bar{x}$.
A~sentence is a formula without free variables.
By $\restr{\mathcal{A}}{\Delta}$ we denote the substructure of the structure $\mathcal{A}$ restricted to the set $\Delta \subseteq A$.
Given a $\Sigma$-structure $\mathcal{A}$ and a relational vocabulary $\tau \subseteq \Sigma$, by the \emph{$\tau$-reduct} of $\mathcal{A}$ we mean the unique $\tau$-structure $\mathcal{B}$ such that $A = B$ and for every predicate $R$ from $\tau$ we have $R^{\mathcal{A}} = R^{\mathcal{B}}$.

\paragraph{Order-invariance}
A sentence $\formule\in\FOtwo(\vocab\cup\{\ordgen\})$, where \ordgen is a binary
relation symbol not belonging to \vocab, is said to be \emph{order-invariant} if for every finite
\vocab-structure \structgen, and every pair of strict linear orders \ord and
\ordbis on \structgendom, $(\structgen,\ord)\models\formule$ if and only if $(\structgen,\ordbis)\models\formule$.
It is then convenient to omit the interpretation for
the symbol \ordgen, and to write $\structgen\models\formule$ if $(\structgen,\ordgen)\models\formule$ for any (or, equivalently, every)
linear order \ordgen.
Note that $\varphi$ is \emph{not} order-invariant if there is a structure~$\mathcal{A}$ and two linear orders $<_0, <_1$ on $A$ such that $(\mathcal{A}, <_0) \models \varphi$ and $(\mathcal{A}, <_1) \not\models \varphi$. 
The set of order-invariant sentences using two variables is denoted \oifotwo.
While determining whether an $\FO$-sentence is order-invariant is undecidable~\cite[Ex.~9.3]{Libkin04}, the situation improves when we allow only two variables: checking order-invariance for $\FOt$-formulae was shown to be in $\coNtwoET$ in~\cite[Thm.~12]{ZeumeH16}.\footnote{The authors of~\cite{ZeumeH16} incorrectly stated the complexity in their Thm. 12, mistaking ``invariance'' with ``non-invariance''.}

\paragraph*{Decision problems}\label{para:intro-decision-problems}
The \emph{finite satisfiability} \emph{problem} for a logic $\logic$ asks whether an input sentence $\varphi$ from $\logic$ is satisfied in some finite structure.
The \emph{finite satisfiability} \emph{problem} for a logic $\logic$ asks whether an input sentence $\varphi$ from $\logic$ is satisfied by all finite structure.
Both the finite satisfiability and validity problems for $\FO$ are undecidable~\cite{turing1938computable,trakhtenbrot1950impossibility}, while for $\FOt$ they are $\NET$-complete and $\coNET$-complete, see~\cite[Thm.~5.3]{GradelKV97} and~\cite[Thm.~3]{Furer83}, respectively.
Note that $\varphi$ is finitely valid iff $\neg \varphi$ is finitely~unsatisfiable.

\paragraph{Definability and similarity}
Let \logic, \logicbis be two logics defined over the same signature, and~\classe be a class of finite structures on this signature.
We say that a property $\propriete\subseteq\classe$ is \emph{definable} (or expressible) in \logic on \classe if there exists an \logic-sentence \formule such that $\propriete=\{\structgen\in\classe:\ \structgen\models\formule\}$. When \classe is the class of all finite structures, we omit it.
We say that $\logic\subseteq\logicbis$ on \classe if every property on \classe definable in \logic is also definable in \logicbis.
Since a sentence which does not mention the linear order predicate is trivially order-invariant, we get the inclusion $\FOtwo\subseteq\ \oifotwo$.
This inclusion is strict~\cite[Example 2]{ZeumeH16}.

The \emph{quantifier rank} of a formula is the maximal number of quantifiers in a branch of its syntactic tree.
Given two \vocab-structures \struct and \structbis, and \logic being one of \FO, \FOtwo and \oifotwo, we write $\struct\logeq\structbis$ if \struct and \structbis satisfy the same \logic-sentences of quantifier rank at most $k$. In this case, we say that \struct and \structbis are \emph{\logic-similar at depth $k$}.

We write $\struct\simeq\structbis$ if \struct and \structbis are isomorphic.

\paragraph*{Atomic types}
An (atomic) \emph{$1$-type} over $\Sigma$ is a maximal satisfiable set of atoms or negated atoms over $\Sigma$ with a free variable~$x$. 
Similarly, an (atomic) \emph{$2$-type} over $\Sigma$ is a maximal satisfiable set of atoms or negated atoms with variables~$x,y$. We stress that it is not necessary that all variables are mentioned in an atom.
Note that the total number of atomic $1$- and $2$-types over $\tau$ is bounded exponentially in~$|\Sigma|$. 
We often identify a type with the conjunction of all its elements. 
The sets of $1$-types and $2$-types over the signature consisting of the symbols appearing in $\varphi$ are denoted by $\AAA_\varphi$ and $\BBB_\varphi$, respectively.
Given a structure $\mathcal{A}$ and an element~$\elem \in A$ we say that $\elem$ \emph{realises} a $1$-type $\alpha$ if $\alpha$ is the unique $1$-type such that $\mathcal{A} \models \alpha[\elem]$ ($\elem$ satisfies $\alpha$ in $\mathcal{A}$).
We then write $\tpgen[\elem]$ to refer to this type.
Similarly, for (non-necessarily distinct) $\elem, \elembis \in A$, we denote by $\tpgen[\elem,\elembis]$ the unique $2$-type \emph{realised} by the pair $(\elem, \elembis)$, \ie the $2$-type $\beta$ such that~$\mathcal{A} \models \beta[\elem, \elembis]$  (the pair $\elem, \elembis$ satisfies $\beta$ in $\mathcal{A}$).
Finally, given a linearly ordered $\vocab$-structure $\orderstructgen$, we split $\tpgen[\elem,\elembis][\orderstructgen]$ into $\tpgen[\elem,\elembis]$ and $\ordtpgen$, where $\ordtpgen$ contains (and, over constant-free vocabularies, is fully determined by) one of the following atoms ``$x<y$'', ``$x>y$''~and~``$x=y$''.

\paragraph{The two-pebble \EF game}
A common way to prove similarity (with respect to some logic) between two structures is to use \EF games. Such a game is played between two players: one tries to highlight the differences between the two structures, while the other attempts to show that they are identical. There exists a whole family of those games, each characterizing a specific logic, in the sense that the second player wins exactly when the two structures cannot be distinguished by the corresponding logic. These games differ by the kind of moves allowed, which reflect the particulars of the logic they characterize.

As we are interested in two-variable first-order logic, let us define the two-pebble \FO game, which captures exactly the expressive power of \FOtwo. This flavor of \EF game has been introduced by Immerman and Kozen~\cite{DBLP:journals/iandc/ImmermanK89}.

Let us consider two structures \struct and \structbis.
The \emph{$k$-round two-pebble \EF game on \struct and \structbis} is played by two players: the spoiler and the duplicator. The spoiler tries to expose differences between the two structures, while the duplicator tries to establish their indistinguishability.

There are two pebbles associated with each structure: $\pebbleone$ and $\pebbletwo$ on \struct, and~$\pebbleonebis$ and~$\pebbletwobis$ on \structbis. Formally, these pebbles can be seen as the interpretations in each structure of two new constant symbols, but it will be convenient to see them as moving~pieces.

At the start of the game, the duplicator places $\pebbleone$ and $\pebbletwo$ on elements of \struct, and~$\pebbleonebis$ and $\pebbletwobis$ on elements of \structbis. The spoiler wins if the duplicator is unable to ensure that $\vocabtp=\vocabtpbis$. Otherwise, the proper game starts. Note that in the usual definition of the starting position, the pebbles are not on the board; however, it will be convenient to have them placed in order to uniformize our invariants in Section~\ref{sec:EFproof}. This change is not profound and does not affect the properties of the game.

For each of the $k$ rounds, the spoiler starts by choosing a structure and a pebble in this structure, and places this pebble on an element of the chosen structure. In turn, the duplicator must place the corresponding pebble in the other structure on an element of that structure. The spoiler wins at once if $\vocabtp\neq\vocabtpbis$. Otherwise, another round is played.
If the spoiler has not won after $k$ rounds, then the duplicator wins.

The main interest of these games is that they capture the expressive power of \FOtwo. We will only need the fact that these games are correct:

\begin{thm}\cite[Lemma 12]{DBLP:journals/iandc/ImmermanK89}
  \label{th:EF}
  If the duplicator has a winning strategy in the $k$-round two-pebble \EF game on \struct and \structbis, then $\struct\fotwoeq\structbis\,.$
\end{thm}

% \paragraph*{Atomic types}
% Let \elem be an element of a structure \structgen. The \emph{atomic type} $\tpgen$ of \elem in \structgen is the set of atomic formulas \formule with at most one free variable \var such that $\structgen,\var\mapsto\elem\models\formule$.

% We define similarly the atomic type $\tpgen[\elem,\elembis]$ of a pair $(\elem,\elembis)$ of elements of $\structgen$ as the set of atomic formulas $\formule$ with free variables in $\{\var,\varbis\}$ such that $\structgen,\var\mapsto\elem,\varbis\mapsto\elembis\models\formule$.

\paragraph*{Gaifman graphs and degree}
The \emph{Gaifman graph} $\gaifman$ of a structure
\structgen is the simple graph with vertices in $A$ and undirected edges between any pair of distinct elements that appear in the same tuple of some relation of $\structgen$. 
% Notice that if a graph is seen as a structure on the signature consisting of a single binary relation symbol, its
% Gaifman graph is none other than the unoriented version of itself.
By $\dist{\structgen}{\elem}{\elembis}$ we denote the distance between \elem and \elembis in
$\gaifman$, defined in the usual way. The distance between an element \elem and a set $B$ of elements is defined as $\min_{\elembis\in B}(\dist{\structgen}{\elem}{\elembis})$.
For $B\subseteq\structgendom$, we let $\Neighgen{B}$ denote the set of elements at distance exactly $1$ from $B$ in $\gaifman$. In particular, $B\cap\Neighgen{B}=\emptyset$.
The \emph{degree} of \structgen is the
maximal degree of its Gaifman graph. 
The class \classe of \vocab-structures is said to have \emph{bounded degree} if there
exists some $d\in\N$ such that the degree of every
$\structgen\in\classe$ is at most $d$.

%!TEX root = ../order-invariant-FO2.tex
\section{Complexity of the invariance problem}\label{sec:complexity}

We study the complexity of the problem of deciding if an input sentence $\varphi$ of $\FO^2$ is order-invariant and establish its $\coNET$-completeness.
We start from a lower bound, inspired by slides of Schweikardt~\cite[Slide~9]{Schweikardt13Slides}. 
The following lemma is crucial:

\begin{lem}
Fix an $\FOt$ sentence $\varphi$, and let $\psi_{<} \coloneqq \exists{x} \left( P(x) \land \forall{y} (y \leq x) \right)$ contain a fresh unary predicate $P$.
Then $(\neg \varphi) \to \psi_<$ is $<$-invariant if and only if $\varphi$ is satisfied by all finite structures with at least two elements.
\end{lem}
\begin{proof}
  Suppose that $\varphi$ is satisfied by every structure with at least two elements.
  To prove that $(\neg \varphi) \to \psi_<$ is $<$-invariant, we establish that $(\neg \varphi) \to \psi_<$ is equivalent to $\top$ over all structures with at least two elements.
  Indeed, over such structures $\neg \varphi$ is equivalent to $\bot$ (and hence $(\neg \varphi) \to \psi_<$ is logically equivalent to $\top$). Thus, $(\neg \varphi) \to \psi_<$ is trivially $<$-invariant.
  For the other direction, suppose that there exists an $\mathcal{A}$ with at least two elements (say, $\elem_1, \elem_2, \ldots, \elem_n$) that satisfies $\neg \varphi$.
  Let $\mathcal{B}$ be the finite structure with $P^\mathcal{B} = \{ \elem_n \}$ whose reduct is $\mathcal{A}$.
  Clearly $\mathcal{B} \models \neg \varphi$, since the predicate $P$ is fresh.
  Finally, let $<_1, <_2$ be two linear orders on $B$: the first one orders the elements $\elem_i$ of $B$ according to their subscript numbering, while the second ordering is like the first one but swaps the positions of $\elem_1$ and $\elem_n$.
  Now note that $(\mathcal{B}, <_1) \models (\neg \varphi) \to \psi_<$ but $(\mathcal{B}, <_2) \not\models (\neg \varphi) \to \psi_<$ (the formula $\psi_<$ basically states that the maximal element of $<_i^{\mathcal{B}}$ satisfies $P$).
  Hence, the formula $(\neg \varphi) \to \psi_<$ is not $<$-invariant.\qedhere

\end{proof}

The above lemma provides a reduction from finite $\FOt$-validity (over structures of size at least two) to checking order-invariance of $\FOt$-sentences.
By $\coNET$-hardness\footnote{The finite validity problem is equally hard to the validity problem over structures of size at least $2$. Let $\varphi$ be an $\FO^2$-sentence, $P$ be a fresh unary predicate, and let $\varphi^P$ be the $P$-relativisation of $\varphi$, i.e. the formula obtained from $\varphi$ by replacing all its subformulae of the form $\exists{x}.\psi$ with $\exists{x}(P(x) \land \psi(x))$ and $\forall{x}.\psi$ with $\forall{x}(P(x) \to \psi(x))$ (and analogously for the variable $y$). One can verify that $\varphi$ is valid if and only if $(\exists{x}.P(x) \land \exists{x}.\neg P(x)) \to \varphi^P$ is satisfied in every structure of size at least $2$.}  of finite validity problem for $\FOt$~\cite[Thm.~3]{Furer83}, we~conclude:

\begin{thm}\label{thm:correctness-of-the-reduction}
    Checking whether an $\FOt$ formula is order-invariant is $\coNET$-hard.
\end{thm}

\noindent Our upper bound uses the following fact, immediate~from~the~definition~of~order-invariance.
\begin{fact}\label{fact:from-inv-to-sat}
An $\FOt$ sentence $\varphi$ is \emph{not} order-invariant iff the sentence $\varphi[{<}{/}{<_0}] \land \neg \varphi[{<}{/}{<_1}]$ is finitely satisfiable over structures interpreting $<_0$ and $<_1$ as linear orders over the domain.
\end{fact}

Let $\FOt_-[<_0, <_1]$ be composed of sentences of the shape $\varphi[{<}{/}{<_0}] \land \neg \varphi[{<}{/}{<_1}]$ for all sentences $\varphi \in \FOt(\Sigma \cup \{ < \})$ over all signatures $\Sigma$. 
To simplify the reasoning about such formulae, we rewrite them into the so-called \emph{Scott Normal Form}.

\begin{defi}
  We say that an $\FOt_-[<_0, <_1]$-sentence is in \emph{Scott's normal form} if
  \begin{itemize}\itemsep0em
    \item its signature is constant-free and employs only predicates of arity $1$ or $2$, and 
    \item it has the form \[ \bigwedge_{i=0}^{1} \left( \forall{x}\forall{y} \; \chi_i(x, y) \land \bigwedge\limits_{j=1}^{m_i} \forall{x}\exists{y} \ \gamma_i^j(x, y) \right), \]
    where the decorated $\chi$ and $\gamma$ are quantifier-free and for each $i \in \{ 0, 1 \}$ the symbol $<_i$ does not appear in $\chi_{1{-}i}$.
  \end{itemize}
\end{defi}

The following lemma justifies the use of formulae in Scott's normal form.

\begin{lem}\label{lemma:Scott-normal-form}
  Given an $\FOt_-[<_0, <_1]$-sentence $\varphi$ one can compute in time polynomial in $|\varphi|$ an $\FOt_-[<_0, <_1]$-sentence $\varphi_{\textit{scott}}$ in Scott's normal form (which is of size polynomial in  $|\varphi|$) such that $\varphi$ is finitely satisfiable over structures interpreting $<_0$ and $<_1$ as linear orders over the domain if and only if $\varphi_{\textit{scott}}$ is. 
\end{lem}
\begin{proof}[Proof sketch.]
  As the translation we are going to provide is folklore and appears almost in every paper regarding $\FOt$, we only sketch its proof here.
  As usual in the context of two-variable logic, we can confirm our attention only to vocabularies composed solely of predicates of arity one and two. A pedantic proof of this fact can be found in the recent textbook of Ian Pratt-Hartmann~\cite[Sec.~3, Lemma~3.4]{PrattHartmann23}.\footnote{Pratt-Hartmann's proof leaves nullary symbols in the vocabulary, but one can easily get rid of them. More precisely, given a nullary symbol $p$ and a formula $\varphi$, we replace $\varphi$ with $(\forall{x} P(x) \lor \forall{x} \neg P(x)) \land \varphi[p / \forall{x} P(x)]$.}
  The reminder of the proof goes via a routine renaming by replacing formulae with a single free variable with a fresh unary predicate and sharpening its interpretation. 
  To do so, we one-by-one take the maximally nested subformula of $\varphi$ starting from a single quantifier, say $\exists{x} \lambda(x,y)$. 
  We then replace the formula $\exists{x} \lambda(x,y)$ with $A(y)$ for a fresh unary predicate $A$, and append the conjunct $\forall{y}\ (\exists{x} \lambda(x,y)) \leftrightarrow A(y)$ to $\varphi$. Note that $\lambda$ does not simultaneously contain $<_0$ and $<_1$. 
  Finally, we split the equivalence into two implications, yielding (by shifting atoms) the formulae of the form $\forall\forall \theta$ and $\forall\exists \xi$. The resulting formula is clearly equisatisfiable to the original one. By repeating this procedure and regrouping and reordering the obtained formulae (and adding vacuous quantifiers if needed), we obtain the desired equivalent formula in normal form.
\end{proof}

For brevity we speak about $\forall\forall$ and $\forall\exists$-conjuncts of $\varphi$ in the normal form given above, meaning the appropriate subformulae of $\varphi$ starting from the mentioned prefix of quantifiers.
Given a model $\mathcal{A} \models \varphi$ of a formula $\varphi$ in Scott's normal form given above and elements $\elem, \elembis \in A$ such that $\mathcal{A} \models \gamma_i^j(\elem, \elembis)$, we call $\elembis$ a $\gamma_i^j$-\emph{witness} for $\elem$ (or simply a witness).
The core of our upper bound proof is the following small model theorem, employing the circular witnessing scheme by Grädel, Kolaitis, and Vardi~\cite[Thm.~4.3]{GradelKV97}. 

\begin{lem}\label{lemma:small-model-property}
If an $\FOt_-[<_0, <_1]$ sentence $\varphi$ (in normal form of Lemma~\ref{lemma:Scott-normal-form}) has a finite model (that interprets $<
_0, <_1$ as linear orders over the domain) then there exists such a finite model with $\mathcal{O}(|\varphi|^3 \cdot 2^{|\varphi|})$ elements.
\end{lem}
\begin{proof}
  Let $M \coloneq m_0 + m_1$ be the total number of $\forall\exists$-conjuncts of $\varphi$, and let $\AAA_{\varphi}$ be the set of all $1$-types over the set of unary symbols of $\varphi$. Recall that $M$ is polynomial in $|\varphi|$ and $|\AAA_{\varphi}|$ is exponential in $|\varphi|$. 
  Suppose that $\mathcal{A}$ is a finite model of $\varphi$ that interprets $<_0, <_1$ as linear orders over the domain, and let us explain our proof plan.

  \textbf{Intuition and proof plan.} In what follows, we employ a variant of the circular witnessing scheme by Grädel, Kolaitis, and Vardi~\cite[Thm.~4.3]{GradelKV97} in order to distinguish three disjoint sets of domain elements of $\mathcal{A}$, namely $W_0, W_1, W_2$ whose union $W \coloneq W_0 \cup W_1 \cup W_2$ will serve as a ``base'' of a new ``small'' model of $\varphi$.
  In order to satisfy $\varphi$, the universal $(\forall\forall)$ and the existential $(\forall\exists)$ conjuncts of $\varphi$ must be fulfilled.
  In our construction, we first distinguish a set $W_0$ of elements of $\mathcal{A}$ that contains elements of $\mathcal{A}$ with $1$-types that do not appear ``frequently'' in $\mathcal{A}$ together with the first and last $M$ realisations w.r.t. $<_0$ and w.r.t. $<_1$ of each other (\ie frequent) $1$-type needed later in the construction.
  Unfortunately, $\mathcal{A}$ restricted to $W_0$ may not be a model of $\varphi$. Indeed, while such a structure satisfies the universal part of $\varphi$, some elements from $W_0$ may violate $\forall\exists$-conjuncts of $\varphi$. 
  To ``repair'' the resulting structure, we construct the set $W_1$ by taking the minimal number of elements in $\mathcal{A}$ so that all elements from $W_0$ have all the required $\gamma_i^j$-witnesses in $\mathcal{A}$ restricted to $W_0\cup W_1$. Note that $W_1$ has no more than $M \cdot |W_0|$ elements.
  The situation is a bit better now: the $\mathcal{A}$  restricted to $W_0\cup W_1$ again satisfies the $\forall\forall$-subformula of $\varphi$ (as the universal formulae are preserved under taking substructures), all elements from $W_0$ have their required witnesses for subformulae $\gamma_i^j$, but there may be still some elements $W_1$ may violate $\forall\exists$-conjuncts of $\varphi$. 
  We ``repair'' the resulting structure one final time by constructing the set $W_2$ analogously to the case of $W_1$, so that all the required $\gamma_i^j$-witnesses for the elements from $W_1$ in $\mathcal{A}$ are in the set $W_0 \cup W_1 \cup W_2$. 
  Unfortunately, for the same reason as before, some elements from $W_2$ may violate some $\forall\exists$-subformulae of $\varphi$. 
  To fix the problem, for each of such elements $\elem$, we suitably select $M$ pairwise-different elements $\elembis_1, \elembis_2, \ldots, \elembis_M$ from $W_0$ and alter binary relations between $\elem$ and all of $\elembis_i$ in order to provide suitable witnesses for $\elem$.
  By our choice of $W_0$ such \emph{different} elements will always exist.
  The resulting $2$-type between $\elem$ and other $\elembis_i$ will be one of the $2$-types already present in $\mathcal{A}$, which ensures that we do not accidentally violate the $\forall\forall$-subformulae of $\varphi$.\footnote{It would be unwise to apply the same approach to $W_1$, as such elements are already constrained by the construction of $W_0$ (while $W_2$ are not!).} Finally, by construction, all the selected elements have the required witnesses in $W_0\cup W_1 \cup W_2$, and hence the $\forall\exists$-subformulae of $\varphi$ are also satisfied in the altered structure $\mathcal{A}$ restricted to $W_0\cup W_1 \cup W_2$. This yields a desired ``small'' model of the formula $\varphi$ and concludes the proof.

  \textbf{The construction of $W_0$.}
  Call a $1$-type $\alpha \in \AAA_{\varphi}$ \emph{rare} if it is realised by at most $4M$ elements in~$\mathcal{A}$ (the factor $4$ comes from the fact that we consider maximal and minimal elements w.r.t. $<_0$ and $<_1$ and $M$ is the number of $\forall\exists$-conjuncts of $\varphi$).
  For every $1$-type $\alpha \in \AAA_{\varphi}$ that is not rare, we select $M$ minimal and $M$ minimal elements w.r.t. each of the orders $<_0$ and $<_1$ that realise $\alpha$. 
  Let $W_0$ be the set of the selected elements from $\mathcal{A}$ together with all elements of $\mathcal{A}$ realising rare $1$-types.

\textbf{The construction of $W_1$ and $W_2$.}
We next close $W_0$ twice under taking witnesses.
More precisely, let $W_1$ be any $\subseteq$-minimal subset of $A$ so that all elements from $W_0$ have all the required $\gamma_i^j$-witnesses in $W_0 \cup W_1$.
Similarly, we define $W_2$ to be any $\subseteq$-minimal subset of $A$ so that all elements from $W_0 \cup W_1$ have all the required $\gamma_i^j$-witnesses in $W_0 \cup W_1 \cup W_2$.
As they are only $M$ different $\gamma_i^j$-formulae in $\varphi$, we have that:
\[ |W_1| \leq M |W_0| \leq  4M^2 |\AAA_{\varphi}|
  \; \text{and} \; |W_2| \leq M |W_0 \cup W_1| \leq M \cdot (4M + 4M^2) |\AAA_{\varphi}| \leq 8M^3  |\AAA_{\varphi}|. 
\]

\textbf{The construction of $\mathcal{B}$.}
Consider the structure $\mathcal{B} \coloneq \restr{\mathcal{A}}{W_0 \cup W_1 \cup W_2}$.
We see that:
\[
|B| \leq |W_0| + |W_1| + |W_2| \leq 
(4M + 4M^2 + 8 M^3) | \cdot |\AAA_{\varphi}| \leq 16M^3 \cdot  |\AAA_\varphi| \leq 16\ |\varphi|^3 \cdot 2^{|\varphi|}.
\]
Hence, the size of $\mathcal{B}$ is as desired.
Note that universal formulae are preserved under substructures, thus $<_0^{\mathcal{B}}, <_1^{\mathcal{B}}$ (defined as $<_0^{\mathcal{A}} \cap B$ and $<_1^{\mathcal{A}} \cap B$) are linear orders over $B$ and $\mathcal{B}$ satisfies the $\forall\forall$-conjuncts of $\varphi$.
Hence, by the construction of $W_1$ and $W_2$, the only reason for $\mathcal{B}$ to not be a model of $\varphi$ is the lack of required $\gamma_i^j$-witnesses for elements from the set $W_2$. 
We fix this issue by reinterpreting binary relational symbols between the sets $W_2$ and $W_0$.

% \textbf{Preparation for repairing the modelhood of $\mathcal{B}$.}
% Before we start, we are going to collect, for each non-rare $1$-type $\alpha$, pairwise-disjoint sets of $M$ minimal and $M$ maximal (w.r.t. each of $<_0^{\mathcal{A}}$, $<_1^{\mathcal{A}}$) realisations of $\alpha$ from $W_0$. 
% The formal definition comes next.
% Fix a non-rare  $\alpha \in \AAA_{\varphi}$.
% Let the set $V_\alpha^{0}$ be composed of the first $M$ $<_0$-minimal elements from $\restr{\mathcal{A}}{W_0}$.
% Next, let $V_\alpha^{1}$ be composed of the last $M$ $<_0$-maximal elements from $\restr{\mathcal{A}}{W_0}$.
% Similarly, let $V_\alpha^{2}$ be composed of the first $M$ $<_1$-minimal elements from $\restr{\mathcal{A}}{W_1 \setminus (V_\alpha^{0} \cup V_\alpha^{1})}$.
% Finally let $V_\alpha^{3}$ be composed of the last $M$ $<_1$-maximal elements~from~$\restr{\mathcal{A}}{W_1 \setminus (V_\alpha^{0} \cup V_\alpha^{1} \cup V_\alpha^{2})}$.
% Put $V_\alpha \coloneqq \bigcup_{k=0}^{3} V_\alpha^{k}$. 
% Notice that all the components of $V_\alpha$ are pairwise disjoint (by construction), and they are well-defined since we included sufficiently many elements~in~$W_1$.

\textbf{Reinterpreting binary relations in $\mathcal{B}$.}
Let us fix an element $\elem$ from $W_2$ that violate some of the $\forall\exists$-conjuncts of $\varphi$ in $\mathcal{B}$, as well as a $\forall\exists$-conjunct~$\psi \coloneq \forall{x}\exists{y}\ \gamma_i^j(x,y)$ whose satisfaction is violated~by~$\elem$. We stress that there are at most $M$ such conjuncts, and $\gamma_i^j$ does not use the $<_{1{-}i}$ predicate.
Since~$\mathcal{A} \models \varphi$ we know that there is an element $\elembis \in A$ (different from $\elem$) such that $\elembis$ is a $\gamma_i^j$-witness for $\elem$ and~$\gamma_i^j$ in~$\mathcal{A}$. Let $\alpha$ be the $1$-type of $\elembis$ in $\mathcal{A}$.
Note that $\alpha$ is not rare. Indeed, otherwise~$\elembis \in W_0$, and hence~$\elembis \in B$. 
Moreover either $\elembis <_i^{\mathcal{A}} \elem$ or $\elem <_i^{\mathcal{A}} \elembis$. 
Suppose that $\elembis <_i^{\mathcal{A}} \elem$ (the other case is analogous), and pick an element $\elem'$ from $W_0$ among $M$-minimal (w.r.t. $<_i$) realisations of the $1$-type $\alpha$. In case $\elem'$ was already selected to provide a missing witness for $\elem$, we take yet another one. 
Note that such an element exists due to the fact that there are only $M$ $\forall\exists$-conjuncts in $\varphi$ and we selected $M$ minimal and maximal realisations of every non-rare $1$-type for each of the orders $<_0$ and $<_1$.
We now alter the binary relations between $\elem$ and $\elem'$ in $\mathcal{B}$ so that the equality 
$\text{tp}_{\mathcal{A}}(\elem, \elembis) = \text{tp}_{\mathcal{B}}(\elem, \elembis)$ holds 
(which can be done as $\elem'$ and $\elembis$ have equal $1$-type).
We repeat the process to provide all missing witnesses for $\elem$.
Observe that all elements from $B$ that had $\gamma_i^j$-witnesses before our redefinition of certain $2$-types, still do have them (as we did not touch $2$-types between them and their witnesses), $\mathcal{B}$ still satisfies the $\forall\forall$-component of~$\varphi$ (since the modified $2$-type does not violate $\varphi$ in $\mathcal{A}$ it does not violate $\varphi$ in $\mathcal{B}$) and $\elem$ has all required witnesses. 
By repeating the strategy for all the other elements from $W_2$ violating $\varphi$, we obtain the desired ``small'' model of $\varphi$. 
\qedhere
\end{proof}

Lemma~\ref{lemma:small-model-property} yields an $\NET$ algorithm for deciding satisfiability of $\FOt_-[<_0, <_1]$ formulae: convert an input into Scott's normal form, guess its exponential size model and verify the modelhood with a standard model-checking algorithm (in $\PT$~\cite[Prop. 4.1]{GradelO99}). 
After applying Fact~\ref{fact:from-inv-to-sat} and Theorem~\ref{thm:correctness-of-the-reduction} we~conclude:
\begin{thm}\label{thm:complexity}
Checking if an $\FOt$-formula is order-invariant is $\coNET$-complete.
\end{thm}

%!TEX root = ../order-invariant-FO2.tex

\newcommand{\phieven}{\varphi_{\textrm{even}}}
\newcommand{\phidendroid}{\varphi_{\textrm{dendroid}}}
\newcommand{\phievenzigzag}{\varphi_{\textrm{even-zig-zag}}}
\newcommand{\ROOT}{\mathtt{ROOT}}
\newcommand{\LS}{\mathtt{LS}}
\newcommand{\RS}{\mathtt{RS}}
\newcommand{\LEAF}{\mathtt{LEAF}}
\newcommand{\SectoRoot}{\mathtt{2nd}}

\section{Can order-invariant $\FOt$ express properties beyond the scope of \FO?}\label{sec:beyond_FO}

While we do not solve the question stated in the heading of this section, we provide a partial solution.
Let $\classe$ be some class of finite structures. 
A sentence $\formule\in\FOtwo(\vocab\cup\{\ordgen\})$, where \ordgen is a binary
relation symbol not belonging to \vocab, is said to be \emph{order-invariant over $\classe$} if for every finite
\vocab-structure \structgen in $\classe$, and every pair of strict linear orders \ord and
\ordbis on \structgendom, $(\structgen,\ord)\models\formule$ iff $(\structgen,\ordbis)\models\formule$.
Note that this is a weakening of the classical condition of order-invariance, and that the usual definition is recovered when $\classe$ is the class of all finite structures.

In what follows, we present a class $\classe$ 
over the vocabulary $\Sigma \coloneqq \{ T, D, S \}$ (for binary predicates $T$, $D$ and $S$) of tree-like finite structures,
and a sentence $\varphi \in \FOt(\Sigma \cup \{ < \})$ ``expressing even depth'' that is
order-invariant over $\classe$ but not equivalent to any first-order sentence over~$\Sigma$.
To do so, we employ \emph{dendroids}, namely finite $\Sigma$-structures $\mathcal{A}$ that, intuitively, are complete directed binary trees decorated with a binary parent-child relation $T^{\mathcal{A}}$, a descendant relation $D^{\mathcal{A}}$, and a sibling relation~$S^{\mathcal{A}}$. Their formal definition comes next.

\begin{defi}
A $\Sigma$-structure $\mathcal{A}$ is a \emph{dendroid} if there is a positive integer $n$ such that
\begin{itemize}
\item $A = \{ 0, 1 \}^{\leq n}$ (\ie the set of all binary words of length at most $n$),
\item $T^{\mathcal{A}} = \{(w,w0), (w,w1) \; \mid \; w \in A, |w| < n \}$, 
\item $D^{\mathcal{A}} = (T^{\mathcal{A}})^+$ (\ie $D^{\mathcal{A}}$ is the transitive closure of $T^{\mathcal{A}}$), and
\item $S^{\mathcal{A}} = \{(w0,w1), (w1,w0) \; \mid \; w \in A, |w| < n \}$.
\end{itemize}
We call the number $n$ the \emph{depth} of $\mathcal{A}$, 
and call the length of a node $v \in A$ the \emph{level} of $v$. 
We also use the terms ``root'' and ``leaf'' in the usual way. 
\end{defi}

The following is a variant of a classical exercise stating that for every threshold $q$ we have that trees that are exponentially deep w.r.t. $q$ are $\FO_q$-equivalent. Indeed:

\begin{lem}\label{prop:EF-games}
  For all $q \in \mathbb{N}$, if $\mathcal{A}, \mathcal{B}$ are dendroids of depth $\geq 2^{q{+}1}$ then $\mathcal{A} \equiv_q^{\FO} \mathcal{B}$.
\end{lem}
\begin{proof}
This is a tedious generalisation of the winning strategy for the duplicator in the $q$-round \EF games on linear orders~\cite[Thm~3.6 Proof \#1]{Libkin04}.
\end{proof}
As an immediate corollary we get:
\begin{cor}\label{cor:even-not-FO}
There is no $\FO(\Sigma)$-formula $\phieven$ such that for every $\mathcal{A} \in \classe$ we have $\mathcal{A} \models \phieven$ iff the depth of $\mathcal{A}$ is even.
\end{cor}

In contrast to Corollary~\ref{cor:even-not-FO}, we show that the even depth query can be defined with an $\FOt(\Sigma \cup \{ < \})$-formula which is order-invariant over $\classe$ (but unfortunately not over the class of all finite structures).
Henceforth we considered \emph{ordered dendroids}, \ie dendroids that are additionally linearly-ordered by $<$. 
Given such an ordered dendroid $\mathcal{T}$, and an element~$\elemter$ with children $\elem, \elembis$ we say that $\elem$ is the \emph{left child} of $\elemter$ iff $\elem <^{\mathcal{T}} \elembis$ holds. 
Otherwise we call $\elem$ the \emph{right child} of~$\elemter$.
A \emph{zig-zag} in $\mathcal{T}$ is a sequence of elements $\elem_0, \elem_1, \ldots, \elem_n$, where $\elem_n$ is a leaf of $\mathcal{T}$, $\elem_0$ is the root of $\mathcal{T}$, $\elem_{2i{+}1}$ is the right child of $\elem_{2i}$ for any $i \geq 0$ and $\elem_{2i}$ is the left child of $\elem_{2i{-}1}$ for any $\frac{n}{2} \geq i \geq 1$. 
A zig-zag is \emph{even} if its last element is the left child of its parent, and \emph{odd} otherwise.
The underlying trees in dendroids are complete and binary,~thus:
\begin{obs}\label{obs:zig-zag-iff-even-depth}
An ordered dendroid $\mathcal{T}$ has an even zig-zag iff $\mathcal{T}$ is of even depth.
Moreover, if $\mathcal{T}$ is a dendroid of even and odd depth then for any linear order $<$ over its domain the ordered dendroid $(\mathcal{T}, <)$ has an even and odd zig-zag, respectively.
\end{obs}
\begin{proof}
Immediate by induction after observing that $\restr{\mathcal{A}}{\{0,1\}^{\leq n}}$, for any positive integer $n$ smaller than the depth of $\mathcal{A}$, is also a dendroid.
\end{proof}

The above lemma suggests that a good way to express the evenness of the depth of a dendroid is to state the existence of an even zig-zag; this is precisely the property that we are going to describe with an $\FOt(\Sigma \cup \{ < \})$-formula.
Let us first introduce a few useful macros:
\[
    \ROOT(x) \coloneqq \neg \exists{y} \; T(y,x) \qquad \LEAF(x) \coloneqq \neg \exists{y} \; T(x,y) \qquad \SectoRoot(x) \coloneqq \exists{y} \; T(y,x) \land \ROOT(y)
\]
\[
    \LS(x) \coloneqq \exists{y} \; S(x,y) \land x < y \qquad \RS(x) \coloneqq \exists{y} \; S(x,y) \land y < x
\]
The first two macros have an obvious meaning. 
The third macro identifies a child of the root, while the last two macros identify the left and the right siblings (according to the linear order $<$), respectively.
Our desired formula $\phievenzigzag$ is then
\[
    \exists{x}\; \Big(  \; \; \; \; [\LEAF(x) \land \LS(x)] 
    \land [\forall{y}\; (\SectoRoot(y) \land D(y, x)) \to \RS(y)] 
\]
\[
  \land [\forall{y}\; (\neg\ROOT(y) \land \neg \SectoRoot(y) \land D(y, x) \land \RS(y)) \to \exists{x}\; T(x,y) \land \LS(x)]
\]
\[
  \land [\forall{y}\; (\neg\ROOT(y) \land \neg \SectoRoot(y) \land D(y, x) \land \LS(y)) \to \exists{x}\; T(x,y) \land \RS(x)] \; \; \; \; \Big).
\]
Note that the above formula, by fixing a leaf, fixes the whole path from such a leaf to the root (since root-to-leaf paths in trees are unique).
To say that such a path is an even zig-zag, we need a base of induction (the first line) stating that the selected leaf is a left child and the root's child lying on this path is its right one, as well as an inductive step stating that every left and right child on the path has a parent which is itself a right and left child, respectively, with the obvious exception of the root and its child.
From there, it is easily shown that:
\begin{prop}\label{prop:2}
An ordered dendroid $\mathcal{T}$ satisfies $\phievenzigzag$ iff it has even depth. 
\end{prop}
\begin{proof}
To prove the right-to-left implication, we use Observation~\ref{obs:zig-zag-iff-even-depth} to infer the existence of an even zig-zag $\elem_0, \elem_1, \ldots, \elem_{2n}$ in $\mathcal{T}$. Taking $\elem_{2n}$ as a witness for the existential quantifier in front of~$\phievenzigzag$ and going back to the definition of an even zig-zag, we get $\mathcal{T} \models \phievenzigzag$.
For the other direction, consider a leaf $\elem$ satisfying the properties enforced in $\phievenzigzag$.
There is a unique path $\rho = \elem_0, \elem_1, \ldots, \elem_n = \elem$ from the root of $\mathcal T$ to $\elem$.
The first line of $\phievenzigzag$ guarantees that $\elem_n$ is a left child and $\elem_1$ is a right child.
We then show by induction, relying on the last two lines of $\phievenzigzag$, that for any $i \geq 0$, $\elem_{2i{+}1}$ is the right child of $\elem_{2i}$, and for $i \geq 1$, $\elem_{2i}$ is the left child of $\elem_{2i{-}1}$.
Thus $\rho$ is an even zig-zag. 
By invoking Observation~\ref{obs:zig-zag-iff-even-depth} again, we get that $\mathcal{T}$ has even depth.
\end{proof}

%% Observe that our formula $\phievenzigzag$ is 
%% order-invariant over $\classe$, which is a direct consequence of the two previous statements. Indeed:
%% \begin{cor}
%% For any dendroid $\mathcal{T} \in \classe$ and any linear orders $<_0^{\mathcal{T}}, <_1^{\mathcal{T}}$ on $T$ we have that $(\mathcal{T}, <_0^{\mathcal{T}}) \models \phievenzigzag$ iff $(\mathcal{T}, <_1^{\mathcal{T}}) \models \phievenzigzag$.
%% \end{cor}
%% Hence, $\phievenzigzag$ is an $<$-invariant query over $\classe$ testing evenness of depth over dendroids from $\classe$. 

As a direct consequence of the previous statement, observe that our formula $\phievenzigzag$ is order-invariant over $\classe$: whether an ordered dendroid has even depth only depends on the underlying dendroid, and not on the particulars of its linear order. Recalling Corollary~\ref{cor:even-not-FO}, we conclude the following:

\begin{thm}\label{thm:beyondfo}
There exists a class of finite structures $\classe$ and an $\FOt(\Sigma \cup \{ < \})$-sentence which is order-invariant over $\classe$, but is not equivalent to any $\FO(\Sigma)$~sentence.
\end{thm}

%!TEX root = ../order-invariant-FO2.tex
\section{Expressive power when the degree is bounded}\label{sec:expr_power}

We have seen in the previous section that if we relax the order-invariant constraint (namely, by requiring invariance only on a restricted class of structures), then one is able to define, with two variables, properties that lie beyond the expressive power of \FO. We conjecture that this is still the case when requiring invariance over the class of all finite structures.

In this section, we go the other way, and show that when one considers only classes of bounded degree, then \oifotwo can only express \FO-definable properties. Note that although the class $\classe$ from Section~\ref{sec:beyond_FO} contains tree-like structures, the descendant relation makes this a dense class of structures (as its Gaifman graph contains cliques of arbitrarily large size), and in particular $\classe$ does not have a bounded degree.

\subsection{Overview of the result}
\label{sec:overview}

We give an upper bound to the expressive power of order-invariant \FOtwo when the degree is bounded:

\begin{thm}
  \label{th:main}
  Let \classe be a class of structures of bounded degree on a purely relational vocabulary.
  Then $\oifotwo\subseteq\FO$ on \classe.
\end{thm}

For the remainder of this section, we fix a purely relational vocabulary \vocab, an integer $d$ and a class \classe of \vocab-structures of degree at most $d$. The constraint on \vocab is merely to ease the presentation of the proof. We discuss in Paragraph~\ref{sec:constants} how to extend Theorem~\ref{th:main} to vocabulary containing constants.

Let us now detail the outline of the proof. The technical parts of the proof will be the focus of Sections~\ref{sec:construction} and~\ref{sec:EFproof}.
Our general strategy is to show the existence of a function $f\colon\N\to\N$ such that every formula $\formule\in\ \oifotwo$ of quantifier rank $k$ is equivalent on \classe (\ie satisfied by the same structures of \classe) to an \FO-formula \formulebis of quantifier rank at~most~$f(k)$.

To prove the existence of such an equivalent \FO-formula \formulebis to any order-invariant \FOtwo-formula \formule on \classe, we will pick $f(k)$ large enough so that for any two structures $\struct,\structbis\in\classe$ such that $\struct\foeq[f(k)]\structbis$, we have $\struct\oifotwoeq\structbis$. In other words, we show that the equivalence relation $\foeq[f(k)]$ is more fine-grained than the equivalence relation $\oifotwoeq$. Let us show that this indeed guarantees the existence of \formulebis. By definition of $\oifotwoeq$, the class of structures satisfying a formula~$\formule$ of $\oifotwo$ of quantifier rank $k$ is a union of equivalence classes for the equivalence relation $\oifotwoeq$, whose intersection with \classe is consequently the intersection of \classe with a union of equivalence classes $c_1,\ldots,c_l$ for $\foeq[f(k)]$. It is folklore (see, \eg, \cite[Cor. 3.16]{Libkin04}) that there are only a finite number of equivalence classes for $\foeq[f(k)]$, each of which being definable by an \FO-sentence of quantifier rank $f(k)$. Then the disjunction \formulebis of the sentences defining $c_1,\ldots,c_l$ (which is again an \FO-sentence of quantifier rank $f(k)$) is indeed equivalent to \formule on \classe, thus showing that the addition of a built-in order does not allow to express more first-order properties on \classe, although it may allow more concise descriptions of these properties.

In order to show that $\struct\oifotwoeq\structbis$ for any structures \struct and \structbis of \classe such that $\struct\foeq[f(k)]\structbis$, we will construct in Section~\ref{sec:construction} two particular different orders on \struct and \structbis (denoted \ordgen on both structures, whose domain can assumed to be disjoint), and we will prove in Section~\ref{sec:EFproof} that \begin{equation}\orderstruct\fotwoeq\orderstructbis\,.\label{eq:order_eq}\end{equation}
This concludes the proof, since any sentence of $\oifotwo$ with quantifier rank at most $k$ holds in \struct iff it holds in \orderstruct (by definition of order-invariance), iff it holds in \orderstructbis (by Equation~(\ref{eq:order_eq})), iff it holds in \structbis.

In the process of constructing the linear orders on \struct and \structbis, it will be convenient to sort the elements of these structures based on the isomorphism type of their \emph{neighbourhood} (for $k\in\N$, the $k$-neighbourhood of some element \elem is the substructure induced by the set of elements that are at distance at most $k$ from \elem). This notion is central in finite model theory, and is the basic tool of the locality methods to prove inexpressibility results. For an introduction on the subject, see, \eg, \cite[Chapter 4]{Libkin04}.

Section~\ref{sec:introneigh} is devoted to the notion of neighbourhood: it is defined in Section~\ref{sec:defneigh}, where we recall that, when the degree is bounded, first-order logic is able to count (up to some threshold) the number of occurrences of each isomorphism type of neighbourhood, and nothing more. In Section~\ref{sec:freqneigh}, we establish a dichotomy between neighbourhood types that appear frequently in a structure, and those that have only a few occurrences. This distinction will highlight which elements play a particular role in the structures \struct and \structbis, and should thus be treated with special care when constructing the linear orders. By contrast, we can be much more liberal with the way we treat elements with a frequent neighbourhood type; in particular, we will use them as padding between critical parts of the orders.

Lastly, we show in Section~\ref{sec:counting} that when the degree is bounded, the inclusion of $\oifotwo$ in \FO still holds when one adds counting quantifiers to \FOtwo.

\subsection{Neighbourhoods}
\label{sec:introneigh}

\subsubsection{Definition and properties of neighbourhoods}
\label{sec:defneigh}

Let us start by defining the notion of neighbourhood of an element in a structure.

Let $c$ be a constant symbol, and let $\structgen\in\classe$.
For $k\in\N$ and $\elem\in\structgendom$, the (pointed) \emph{$k$-neighbourhood}
$\neighgen$ of $\elem$ in $\structgen$ is the $(\vocab\cup\{c\})$-structure
whose restriction to the vocabulary~$\vocab$ is the substructure of $\structgen$ induced by the set
$\boulegen \coloneq \{\elembis\in\structgendom:\dist{\structgen}{\elem}{\elembis}\leq k\}\,,$ and where~$c$ is interpreted
as $\elem$. In other words, it consists of all the elements at distance at most $k$ from $\elem$ in $\structgen$, together with the relations they share in $\structgen$; the center $\elem$ being marked by the constant $c$. We sometimes refer to $\boulegen$ as the $k$-neighbourhood of $a$ in $\structgen$ as well, but the context will always make it clear whether we refer to the whole substructure or only its domain.
The \emph{$k$-neighbourhood type} $\type=\ntpgen$ of $\elem$ in $\structgen$ is the
isomorphism class of its $k$-neighbourhood. We say that $\type$ is a $k$-neighbourhood type
over $\vocab$, and that $\elem$ is an \emph{occurrence} of
$\type$. We denote by $\nocc{\structgen}{\type}$ the number of occurrences of
$\type$ in $\structgen$ (\ie the number of elements with $k$-neighbourhood $\type$ in $\structgen$), and we write
$\threq{\struct}{\structbis}{k}{t}$ to mean that for every $k$-neighbourhood type
$\type$, $\nocc{\struct}{\type}$ and $\nocc{\structbis}{\type}$ are
either equal, or both larger than $t$.

Let $\Ntp$ denote the set of $k$-neighbourhood types over $\vocab$ occurring in structures of degree at most $d$. Note that $\Ntp$ is a finite set.

The interest of this notion resides in the fact that when the degree is bounded, \FO is exactly able to count the number of occurrences of neighbourhood types up to some threshold~\cite{DBLP:journals/iandc/FaginSV95}. We will only use one direction of this characterization, namely:

\begin{prop}
  \label{prop:FO_neigh}
  For all integers $k$ and $t$, there exists some $\hat f(k,t)\in\N$ (which also depends on the bound $d$ on the degree of structures in \classe) such that for all structures $\struct,\structbis\in\classe$,
  $\struct\foeq[\hat f(k,t)]\structbis$ implies $\threq{\struct}{\structbis}{k}{t}$.
\end{prop}

\begin{rem}
  \label{rem:end_proof}
  Assume the existence of a function $\thr\colon\N\to\N$ such that whenever $\struct,\structbis\in\classe$ are such that $\threq{\struct}{\structbis}{k}{\thr(k)}$, one can construct \ordgen satisfying Equation~(\ref{eq:order_eq}).
  
  Summing up the above discussion, we show that the existence of such a $\thr$ is enough to end the proof of Theorem~\ref{th:main}. Indeed, Proposition~\ref{prop:FO_neigh} ensures that for $f \colon k \mapsto \hat{f}(k,\thr(k))$, if $\struct\foeq[f(k)]\structbis$ then there exists \ordgen such that $\orderstruct\fotwoeq[k]\orderstructbis$, entailing~$\struct\oifotwoeq[k]\structbis$.

  In other words, this means the equivalence relation $\foeq[f(k)]$ is more fine-grained than $\oifotwoeq[k]$ on \classe. Given an \oifotwo-sentence of quantifier rank $k$, one can decompose it along the equivalence classes of $\oifotwoeq[k]$, and thus, in turn, decompose it according to $\foeq[f(k)]$ (i.e. as a disjunction of \FO-sentences of rank $f(k)$ defining the relevant $\foeq[f(k)]$ classes). This gives a translation on \classe from \oifotwo-sentences of quantifier rank $k$ to \FO-sentences of rank $f(k)$.
\end{rem}

Let us now explain how the function $\thr$ is chosen.

\subsubsection{Frequency of a neighbourhood type}
\label{sec:freqneigh}

Let us write $\nbrtp \coloneqq |\Ntp|$, the number of types of $k$-neighbourhoods possibly appearing in structures of $\classe$ (which have degree at most $d$).

In this section, we show that if we consider the set $\Freq[\structgen]$ of $k$-neighbourhood types that have enough occurrences in some $\structgen\in\classe$ (where ``enough'' will be given a precise meaning later on), each type in $\Freq[\structgen]$ must have many occurrences that are scattered across~$\structgen$. Not only that, but we can also make sure that such occurrences are far from all the occurrences of every $k$-neighbourhood type not in $\Freq[\structgen]$, which by definition have few occurrences in $\structgen$. %At first sight, such a distinction (which will be formalized later on) between rare neighbourhood types and frequent neighbourhood types appears to be circular. However, because the degree is bounded, $\nbrtp$ is bounded too, which allows this definition to be well-founded.

Such a dichotomy was introduced in~\cite{DBLP:journals/lmcs/Grange21}; we simply accommodate this construction to our needs. In the remainder of this section, we recall and adapt this construction.

The following lemma, which heavily relies on the degree boundedness assumption, states that, when given a small set $B$ and a bounded number of large enough sets $C_1,\ldots,C_n$, it is always possible to find many elements in each $C_j$ which are far apart from one another, as well as far from $B$. It is a strengthening of \cite[Lemma 3.2]{nevsetvril2010first}).

\begin{lem}
  \label{lem:scatter}
  Given three positive integers $m$, $\delta$, $s$, there exists a threshold $g(m,\delta,s)\in\N$ such that for all $\structgen\in\classe$, all $B\subseteq\structgendom$ of size at most $s$, and all subsets $C_1,\ldots,C_n\subseteq\structgendom$ (with~$n\leq\nbrtp$) of size at least $g(m,\delta,s)$, it is possible to find elements $c_j^1,\ldots,c_j^m\in C_j$ for all $j\in\{1,\ldots,n\}$, such that for all $j,j'\in\{1,\ldots,n\}$ and $i,i'\in\{1,\ldots,m\}$, $\dist{\structgen}{c_j^i}{B}>\delta$ and $\dist{\structgen}{c_j^i}{c_{j'}^{i'}}>\delta$ if $(j,i)\neq(j',i')$.

  Furthermore, such a function $g$ can be chosen such that:
  \begin{itemize}
  \item for all integers $m,\delta,s$, $g(m,\delta,s)\geq s$, and
  \item $g$ is non-decreasing in its third argument, i.e. for every integers $m,\delta,s,s'$ such that $s\leq s'$, $g(m,\delta,s)\leq g(m,\delta,s')$.
  \end{itemize}
\end{lem}

\begin{proof}
  Let us fix $m,\delta,s\in\N$ and $\structgen\in\classe$, which has degree at most $d$ by assumption.

  Let us start by proving that for any $p\in\N$ and sets of elements $X, Y\subseteq\structgendom$ such that $|X|\leq p$ and $|Y|\geq(p+m)(d^\delta+1)$, there exist elements $y^1,\ldots,y^m\in Y$ such that for all $i,i'\in\{1,\ldots,m\}$, $\dist{\structgen}{y^i}{Y}>\delta$ and $\dist{\structgen}{y^i}{y^{i'}}>\delta$ if $i\neq i'$.

  Let us consider the $\delta$-th power $\powgaifman$ of the Gaifman graph of $\structgen$, whose vertex set is the same as $\gaifman$ and where there is an edge between two vertices iff their distance in $\gaifman$ is at most $\delta$. The degree of $\powgaifman$ can be coarsely bounded by $d^\delta$. Thus, there are at most $p (d^\delta+1)$ elements in or adjacent to $X$ in $\powgaifman$, and at least $m(d^\delta+1)$ elements in $Y\setminus(X\cup\Neighgen{X}[\powgaifman])$.
  Consider an arbitrary $(d^\delta+1)$-coloring of $\powgaifman$ (where two adjacent vertices have different colors), which is guaranteed to exist since the degree of $\powgaifman$ is at most $d^\delta$. By the pigeonhole principle, one can find $m$ elements $y^1,\ldots,y^m$ in $Y\setminus(X\cup\Neighgen{X}[\powgaifman])$ with the same color. Since these elements do not belong to $X\cup\Neighgen{X}[\powgaifman]$, they are at distance at least $\delta+1$ from $X$ in $\gaifman$. Furthermore, they form an independent set in $\powgaifman$ (by definition of a coloring), and are thus at distance at least $(\delta+1)$ from one another in $\gaifman$. This proves the claim.
  
  We are now ready to prove the lemma: let $B$ be a subset of $\structgendom$ of size at most $s$. Define $g_n(m,\delta,s) \coloneqq (s+nm)(d^\delta+1)$. We show by induction on $n$ that for every $C_1,\ldots,C_n\subseteq\structgendom$ of size at least $g_n(m,\delta,s)$, one can find elements $c_j^i\in C_j$, for $j\in\{1,\ldots,n\}$ and $i\in\{1,\ldots,m\}$, which are at distance at least $\delta+1$ from each other and from $B$ in $\gaifman$. For $n=0$, there is nothing to show. For the inductive step from $n$ to $n+1$: since $g_{n+1}(m,\delta,s)\geq g_n(m,\delta,s)$, we start by using the induction hypothesis to get suitable elements $c_j^i\in C_j$ for $j\in\{1,\ldots,n\}$ and $i\in\{1,\ldots,m\}$. We can then rely on the previous claim with $X \coloneqq B\cup\bigcup_{j=1}^n\{c_j^1,\ldots,c_j^m\}$ (\ie $p=s+nm$) and $Y \coloneqq C_{n+1}$ (which, by choice of $g_{n+1}(m,\delta,s)$, satisfy the cardinality condition of the claim), and pick the $y^1,\ldots,y^m$ as our $c^1_{n+1},\ldots,c^m_{n+1}$.
  
  Setting $g(m,\delta,s) \coloneqq g_{\nbrtp}(m,\delta,s)$ concludes the proof: it is routine to check that such a function $g$ satisfies the last two properties of the lemma.
\end{proof}

\begin{rem}
  \label{rem:M}
  Note that the size of (an instance of) a neighbourhood type $\type\in\Ntp$ is at most
  \begin{equation*}
    M_k^d \coloneqq
    \begin{cases}
      2k+1&\text{if }d=2\\
      1+d\frac{(d-1)^k-1}{d-2}&\text{otherwise}      
    \end{cases}
  \end{equation*}
  and thus can be bounded by a function of $k$ (recall that the signature $\vocab$ and the degree $d$ are fixed).
\end{rem}

We now explain how to use Lemma~\ref{lem:scatter} in order to define the notion of $k$-neighbourhood frequency in a structure $\structgen$.
Let us consider
\begin{equation}
  \label{eq:m_and_k}
  m \coloneqq 2(k+1)\cdot M_k^d!\qquad\text{and}\qquad\delta \coloneqq 4k\,.
\end{equation}

Our goal is, given a structure $\structgen\in\classe$, to partition the $k$-neighbourhood types into two categories: the frequent types, and the rare types. The property we wish to enforce is that there exist in $\structgen$ at least $m$ occurrences of each one of the frequent $k$-neighbourhood types which are both
\begin{itemize}
\item at distance greater than $\delta$ from one another, and
\item at distance greater than $\delta$ from every occurrence of a rare $k$-neighbourhood type.
\end{itemize}

To that end, for $t\in\N$ and $\structgen\in\classe$ let $\Freq[\structgen][k][\geq t]\subseteq\Ntp$ denote the set of $k$-neighbourhood types which have at least $t$ occurrences in~$\structgen$.

\begin{lem}
  \label{lem:threshold}
  Let $k\in\N$ and define $m$ and $\delta$ as in Equation~(\ref{eq:m_and_k}). There exists $\thr(k)\in\N$ such that for every $\structgen\in\classe$, there exists some $t\leq\thr(k)$, for which \[t\geq g\Bigg(m,\delta,\sum_{\type\in\Ntp\setminus\Freq[\structgen][k][\geq t]}\nocc{\structgen}{\type}\Bigg)\,.\]
\end{lem}

\begin{proof}

  Define by induction $h_1(k) \coloneqq 0$ and $h_{i+1}(k) \coloneqq h_i(k)+g(m,\delta,h_i(k))$. This is well defined, as $m$ and $\delta$ depend only on $k$. We now show that the lemma holds for \[\thr(k) \coloneqq h_{\nbrtp+2}(k)\,.\]% \[\thr(k) \coloneqq g(m,\delta,h_{\nbrtp+1}(k))+1\,.\]

  Fix some structure $\structgen\in\classe$ and let $\type_1,\ldots,\type_{\nbrtp}$ be an enumeration of $\Ntp$ such that for all $i<j$ we have $\nocc{\structgen}{\type_i}\leq\nocc{\structgen}{\type_j}$.

  Consider the algorithm described in Procedure~\ref{alg:types}. The set $X$, which is empty at first, will eventually contain all types that will be declared rare. In words, this algorithm proceeds as follows:
  \begin{enumerate}
  \item First, we mark every $k$-neighbourhood type as frequent.
  \item\label{enu:while} Among the types which are currently marked as frequent, let $\tau$ be one with the smallest number of occurrences in $\structgen$.
  \item If $\nocc{\structgen}{\tau}$ is at least $g(m,\delta,s)$ ($g$ being the function from Lemma~\ref{lem:scatter}) where $s$ is the total number of occurrences of all the $k$-neighbourhood types which are currently marked as rare, then we are done and the marking frequent/rare is final. Otherwise, mark $\tau$ as rare, and go back to step~\ref{enu:while} if there remains at least one frequent $k$-neighbourhood type
  \end{enumerate}
  
  \begin{algorithm}[!ht]
    \caption{Determining the rare and frequent neighbourhood types}
    \label{alg:types}
    
    $X\gets\emptyset$
    
    $i\gets 1$

    \textbf{while} $i\leq\nbrtp$ and $\nocc{\structgen}{\type_i}\leq g(m,\delta,\sum_{\type\in X}\nocc{\structgen}{\type})$

    \qquad $X\gets X\cup\{\type_i\}$

    \qquad $i$++

    \textbf{end} \tcp*{the rare types are those in $X$; types outside $X$ are frequent types}
  \end{algorithm}

%  We look at the types $\type_i$ by increasing number of their occurrences in $\structgen$. If the number of occurrences of a $\type_i$ is at least $g(m,\delta,p)$, we stop. This ensures that

  Let us establish the following invariant for the ``while'' loop:

  \begin{equation}
    \label{eq:bound_nocc} \sum_{\type\in X}\nocc{\structgen}{\type}\leq h_i(k).
  \end{equation}

  Before the first execution of the body of the loop, both sides of Equation~(\ref{eq:bound_nocc}) are zero. Suppose now that Equation~(\ref{eq:bound_nocc}) is true for some $i\leq\nbrtp$ at the start of the loop, and let us prove that Equation~(\ref{eq:bound_nocc}) still holds after executing the body of the loop. We went from $i$ to $i+1$. Let $X_i$ denote the old value of $X$, and $X_{i+1}$ its new value; in other words, $X_{i+1}=X_i\cup\{\type_i\}$. The left part of Equation~(\ref{eq:bound_nocc}) has been augmented by $\nocc{\structgen}{\type_i}$, which, because the ``while'' condition was satisfied, is less than $g(m,\delta,\sum_{\type\in X_i}\nocc{\structgen}{\type})$. Thus, \[\sum_{\type\in X_{i+1}}\nocc{\structgen}{\type}\leq\sum_{\type\in X_i}\nocc{\structgen}{\type} + g\Bigg(m,\delta,\sum_{\type\in X_i}\nocc{\structgen}{\type}\Bigg)\,.\] By assumption, $\sum_{\type\in X_i}\nocc{\structgen}{\type}\leq h_i(k)$, and since $g$ is non-decreasing in its third variable, we get $\sum_{\type\in X_{i+1}}\nocc{\structgen}{\type}\leq h_{i+1}(k)$, which means that Equation~(\ref{eq:bound_nocc}) is preserved by the body of the loop.

  Let us now look at the two ways we can exit the ``while'' loop, and prove that the statement of the lemma holds whichever the case:

  \begin{itemize}
  \item If we exit the loop because the condition ``$i\leq\nbrtp$'' is false, then we have $i=\nbrtp+1$ and $X=\Ntp$. By Equation~(\ref{eq:bound_nocc}), $\sum_{\type\in \Ntp}\nocc{\structgen}{\type}\leq h_{\nbrtp+1}(k)$, and we can choose $t \coloneqq g(m,\delta,h_{\nbrtp+1}(k))$. We then indeed get $t\leq\thr(k)$ and, because $g$ is non-decreasing in its third argument, \[t\geq g\Bigg(m,\delta,\sum_{\type\in\Ntp}\nocc{\structgen}{\type}\Bigg)=g\Bigg(m,\delta,\sum_{\type\in\Ntp\setminus\Freq[\structgen][k][\geq t]}\nocc{\structgen}{\type}\Bigg)\,.\]

  \item Otherwise, $i\leq\nbrtp$ and it is ``$\nocc{\structgen}{\type_i}\leq g(m,\delta,\sum_{\type\in X}\nocc{\structgen}{\type})$'' which fails. We can then choose $t \coloneqq g(m,\delta,\sum_{\type\in X}\nocc{\structgen}{\type})+1$ (which is, again, smaller than $\thr(k)$). Let us prove that $X$ is exactly the set of $k$-neighbourhood types that have less than $t$ occurrences in $\structgen$. For $\type\in X$, we have $\nocc{\structgen}{\type}\leq \sum_{\type'\in X}\nocc{\structgen}{\type'}\leq g(m,\delta,\sum_{\type'\in X}\nocc{\structgen}{\type'})< t$ (the second inequality follows from the good properties of $g$, cf. Lemma~\ref{lem:scatter}). Conversely, if $\type\notin X$, it means that $\nocc{\structgen}{\type}\geq\nocc{\structgen}{\type_i}$ (recall that the $\type_i$ are examined in number of occurrences in $\structgen$), which in turn is at least $t$.

    Then we have \[t\geq g\Bigg(m,\delta,\sum_{\type\in X}\nocc{\structgen}{\type}\Bigg)= g\Bigg(m,\delta,\sum_{\type\in\Ntp\setminus\Freq[\structgen][k][\geq t]}\nocc{\structgen}{\type}\Bigg)\,. \qedhere\]
  \end{itemize}
\end{proof}

We are now ready to formally establish the notion of frequent type:

\begin{defi}
  \label{def:frequent}
  Fix an integer $k$ and a structure $\structgen\in\classe$, and let $\Freq[\structgen] \coloneqq \Freq[\structgen][k][\geq t]$ for the smallest threshold $t\leq\thr(k)$ such that \[t\geq g\Bigg(m,\delta,\sum_{\type\in\Ntp\setminus\Freq[\structgen][k][\geq t]}\nocc{\structgen}{\type}\Bigg)\,,\] whose existence is guaranteed by Lemma~\ref{lem:threshold}.
  Some $k$-neighbourhood type $\type\in\Ntp$ is said to be \emph{frequent in $\structgen$} if it belongs to $\Freq[\structgen]$; that is, if $\nocc{\structgen}{\type}\geq t$. Otherwise, $\type$ is said to be \emph{rare in $\structgen$}.
\end{defi}

\begin{rem}
  \label{rem:no_freq}
  A structure \structgen without frequent type has size bounded by $\thr(k)\cdot\nbrtp$ since, by Definition~\ref{def:frequent}, each $\type\in\Ntp$ has at most $\thr(k)$ occurrences in \structgen.
\end{rem}

Recall how $m$ and $\delta$ were defined in Equation~(\ref{eq:m_and_k}). Combining Lemmas~\ref{lem:scatter}--\ref{lem:threshold}, we~get:

\begin{prop}
  \label{prop:scatter}
  Let $k\in\N$. In every structure $\structgen\in\classe$, one can find~$2(k+1)\cdot M_k^d!$ occurrences of each frequent $k$-neighbourhood type in $\structgen$ which are at distance at least~$4k+1$ from one another and from the set of occurrences of every rare $k$-neighbourhood type in $\structgen$.
\end{prop}

With these notions of frequent and rare neighbourhood types established, we now turn to the construction of the two linear orders.

\subsection{\texorpdfstring{Constructing linear orders on \struct and \structbis}{Constructing linear orders on A0 and A1}}
\label{sec:construction}

Recall that our goal is to find a function~$f$ such that, given two structures \struct, \structbis in \classe (assumed without loss of generality to have disjoint domains) such that  $\struct\foeq[f(k)]\structbis$, we are able to construct linear orders on \structdom and \structbisdom (both denoted \ordgen) such that $\orderstruct\fotwoeq\orderstructbis$.
In this section, we define $f$ and we detail the construction of such orders. The proof of \FOt-similarity between \orderstruct and \orderstructbis will be the focus of Section~\ref{sec:EFproof}.

For every integer $k$, we choose an $f(k)$ large enough so that
\begin{itemize}
\item for every structures $\struct,\structbis\in\classe$, $\struct\foeq[f(k)]\structbis$ entails $\threq{\struct}{\structbis}{k}{\thr(k)}$, where $\thr(k)$ is defined in Lemma~\ref{lem:threshold},
\item $f(k)\geq \thr(k)\cdot\nbrtp +1$,
\item $f$ satisfies an additional lower bound made explicit later in this section, in the paragraph \emph{Transfer on $\structbis$} below.
\end{itemize}

In what follows, we fix an integer $k\in\N$, as well as two structures $\struct$ and $\structbis$ of $\classe$ such~that
\begin{equation}
  \label{eq:foeq}
  \struct\foeq[f(k)]\structbis\,
\end{equation}
for such a choice of $f$.

By choice of $f$, we get $\threq{\struct}{\structbis}{k}{\thr(k)}$. This ensures, for all $t\leq\thr(k)$, that $\Freq[\struct][k][\geq t]=\Freq[\structbis][k][\geq t]$, and that
 \[\sum_{\type\in\Ntp\setminus\Freq[\struct][k][\geq t]}\nocc{\struct}{\type}=\sum_{\type\in\Ntp\setminus\Freq[\structbis][k][\geq t]}\nocc{\structbis}{\type}\,.\] 
The smallest $t$ such that \[t\geq g\Bigg(m,\delta,\sum_{\type\in\Ntp\setminus\Freq[\struct][k][\geq t]}\nocc{\struct}{\type}\Bigg)\] is therefore also the smallest $t$ such that \[t\geq g\Bigg(m,\delta,\sum_{\type\in\Ntp\setminus\Freq[\structbis][k][\geq t]}\nocc{\structbis}{\type}\Bigg)\,,\] meaning exactly (according to Definition~\ref{def:frequent}) that $\Freq=\Freq[\structbis]$. In other words, the notion of frequent and rare $k$-neighbourhood types coincide on $\struct$ and $\structbis$. From this point on, we thus refer to \emph{frequent} and \emph{rare} $k$-neighbourhood types without mentioning the structure involved.

Furthermore, if $\Freq$ (and thus $\Freq[\structbis]$) is empty, then according to Remark~\ref{rem:no_freq} we have $|\struct|\leq \thr(k)\cdot\nbrtp$. As we chose $f(k)\geq\thr(k)\cdot\nbrtp +1$, Equation~(\ref{eq:foeq}) entails that $\struct\simeq\structbis$, and in particular that $\struct\oifotwoeq\structbis$. From now on, we thus suppose that there is at least one frequent $k$-neighbourhood type.

To construct the two linear orders, we need to define the notion of $k$-environment: given~$\structgen\in\classe$, a linear order $\ordgen$ on $\structgendom$, $k\in\N$ and an element $\elem\in\structgendom$, we define the \emph{$k$-environment} $\Envgen$ \emph{of $\elem$ in $\orderstructgen$} as the restriction of $\orderstructgen$ to the $k$-neighbourhood of $\elem$ in $\structgen$, where $\elem$ is the interpretation of the constant symbol $c$. Note that the order is not taken into account when determining the domain of the substructure (it would otherwise be $\structgendom$, given that any two distinct elements are adjacent for $\ordgen$).
The \emph{$k$-environment type} $\envtpgen$ is the isomorphism class of $\Envgen$.
In other words, $\envtpgen$ contains the information of $\neighgen$ together with the order of its elements in $\orderstructgen$.

Given $\type\in\Ntp$, we define \Ord as the set of $k$-environment types whose underlying $k$-neighbourhood type is $\type$.

For each $\indin$, we aim to partition $\structinddom$ into $2(2k+1)+2$ segments:
\[\structinddom=\Rareind\uplus\biguplus_{j=0}^{2k}(\Leftind\uplus\Rightind)\uplus\Middleind\,.\]
Once we have set a linear order on each segment, the linear order \ordgen on $\structinddom$ will result from the concatenation of the orders on the segments as follows: \[(\structinddom,\ordgen) \coloneqq \Rareind\cdot\Leftind[0]\cdot\Leftind[1]\cdots\Leftind[2k]\cdot\Middleind\cdot\Rightind[2k]\cdots\Rightind[1]\cdot\Rightind[0]\,.\]
Formally, we order the segments as \[\Rareind\ordgen\Leftind[0]\ordgen\Leftind[1]\ordgen\cdots\ordgen\Leftind[2k]\ordgen\Middleind\ordgen\Rightind[2k]\ordgen\cdots\ordgen\Rightind[1]\ordgen\Rightind[0]\,,\] and for $x,y\in\structinddom$, we let $x\ordgen y$ either if $x$'s segment is smaller than $y$'s segment, or if both belong to the same segment and $x$ is less than $y$ according to the order for this segment.

Each segment $\Leftind$, for $j\in\{0,\ldots,2k\}$, is itself decomposed into two segments $\NLeftind$ and $\ULeftind$, and the order on $\Leftind$ will be defined as the concatenation of the order on $\NLeftind$ with the order on $\ULeftind$ (meaning that all elements from $\NLeftind$ will be smaller than all elements from $\ULeftind$). The $\ULeftind$ for $j\in\{k+1,\ldots,2k\}$ will be empty; they are defined solely in order to keep the notations uniform. The 'N' stands for ``neighbour'' and the 'U' for ``universal'', for reasons that will soon become apparent.
Symmetrically, each $\Rightind$ is decomposed into $\URightind\cdot\NRightind$ (and thus all elements of $\URightind$ will be smaller than all elements of $\NRightind$), with $\URightind$ being empty as soon as $j\geq k+1$.

For $\indin$ and $r\in\{0,\ldots,2k\}$, we define $\Segmentind$ as \[\Segmentind \coloneqq \Rareind\cup\bigcup_{j=0}^r(\Leftind\cup\Rightind)\,.\]

Let us now explain how the segments are constructed in \struct; see Figure~\ref{fig:segments} for an illustration.

\begin{figure*}[!ht]
  \centering
  \begin{tikzpicture}[scale=.98]

    \newcommand{\topy}{.5}
    \newcommand{\boty}{.5}

    %N
    \foreach \b/\e in {1/2,3/4,5.5/6.5,7.5/8.5,9/10,11/12,13.5/14.5}
    \draw[pattern=north west lines, pattern color=green!50] (\b,-\boty) rectangle (\e,\topy);

    %U
    \foreach \b/\e in {2/3,4/5,6.5/7.5,12.5/13.5,14.5/15.5}
    \draw[pattern=dots, pattern color=black!20] (\b,-\boty) rectangle (\e,\topy);

    \foreach \b/\e in {0/1,10/11}
    \draw[pattern=crosshatch dots, pattern color=red!50] (\b,-\boty) rectangle (\e,\topy);

    \foreach \a in {5.3,8.8,12.3}
    \node at (\a,0) {\small $\cdots$};

    \node at (.5,0) {\small $\Rare$};
    \node at (1.5,0) {\small $\NLeft[0]$};
    \node at (2.5,0) {\small $\ULeft[0]$};
    \node at (3.5,0) {\small $\NLeft[1]$};
    \node at (4.5,0) {\small $\ULeft[1]$};
    \node at (6,0) {\small $\NLeft[k]$};
    \node at (7,0) {\small $\ULeft[k]$};
    \node at (8,0) {\small $\NLeft[k\raisebox{.25\height}{\scalebox{.5}{+}}1]$};
    \node at (9.5,0) {\small $\NLeft[2k]$};
    \node at (10.5,0) {\small $\Middle$};
    \node at (11.5,0) {\small $\NRight[2k]$};
    \node at (13,0) {\small $\URight[1]$};
    \node at (14,0) {\small $\NRight[1]$};
    \node at (15,0) {\small $\URight[0]$};

    \foreach \b/\e in {
      .66/1.33,1.66/3.33,3.66/5.16,2.66/3.33,4.66/5.16,5.33/5.83,6.16/7.83,7.16/7.83,8.16/8.66,8.83/9.33,9.66/10.33,10.66/11.33,11.66/12.16,12.33/12.83,12.33/13.83,14.16/14.83
      %% .3/1.6,.4/1.4,.6/1.6,
      %% 1.4/3.4,1.6/3.6,
      %% 2.4/3.3,2.4/3.8,2.5/3.5,
      %% 3.3/5.15,3.4/5.2,3.5/5.3,3.8/5.25,
      %% 4.4/5.3,4.4/5.31,4.5/5.22,
      %% 5.4/5.9,5.35/6.2,
      %% 5.9/8,5.9/8.2,6.2/7.8,
      %% 7/7.9,7.3/8.1,
      %% 7.8/8.65,7.9/8.6,8.1/8.7,8/8.62, 8.2/8.7,
      %% 8.8/9.45,8.9/9.45,8.85/9.6,
      %% 9.45/10.75,9.6/10.35, 9.45/10.4,
      %% 10.25/11.4,10.6/11.6,10.7/11.4,
      %% 11.4/12.15, 11.6/12.2,
      %% 12.3/12.9,12.35/13.1,12.38/13.9,12.4/14.2,12.38/14.12,
      %% 14.2/15.1,13.9/14.95,14.12/14.8
    }
    \draw[-] (\b,\boty) edge[out=70,in=110,-] (\e,\topy);

    \draw [decorate,decoration={brace,amplitude=3pt},xshift=0pt,yshift=0pt] (2.9,-0.6) -- (1.1,-0.6) node [black,midway,yshift=-0.4cm] {$\Left[0]$};
    \draw [decorate,decoration={brace,amplitude=3pt},xshift=0pt,yshift=0pt] (4.9,-0.6) -- (3.1,-0.6) node [black,midway,yshift=-0.4cm] {$\Left[1]$};
    \draw [decorate,decoration={brace,amplitude=3pt},xshift=0pt,yshift=0pt] (7.4,-0.6) -- (5.6,-0.6) node [black,midway,yshift=-0.4cm] {$\Left[k]$};
    \draw [decorate,decoration={brace,amplitude=3pt},xshift=0pt,yshift=0pt] (8.4,-0.6) -- (7.6,-0.6) node [black,midway,yshift=-0.4cm] {$\Left[k+1]$};
    \draw [decorate,decoration={brace,amplitude=3pt},xshift=0pt,yshift=0pt] (9.9,-0.6) -- (9.1,-0.6) node [black,midway,yshift=-0.4cm] {$\Left[2k]$};
    \draw [decorate,decoration={brace,amplitude=3pt},xshift=0pt,yshift=0pt] (11.9,-0.6) -- (11.1,-0.6) node [black,midway,yshift=-0.4cm] {$\Right[2k]$};
    \draw [decorate,decoration={brace,amplitude=3pt},xshift=0pt,yshift=0pt] (14.4,-0.6) -- (12.6,-0.6) node [black,midway,yshift=-0.4cm] {$\Right[1]$};
    \draw [decorate,decoration={brace,amplitude=3pt},xshift=0pt,yshift=0pt] (15.4,-0.6) -- (14.6,-0.6) node [black,midway,yshift=-0.4cm] {$\Right[0]$};

  \end{tikzpicture}
  \caption{Elements from two distinct segments can only be adjacent if there is a black curvy edge between those segments. The order \ordgen grows from the left to the right. Note that $\NRight[0]$ is empty and thus left out of the picture.}
  \label{fig:segments}
\end{figure*}
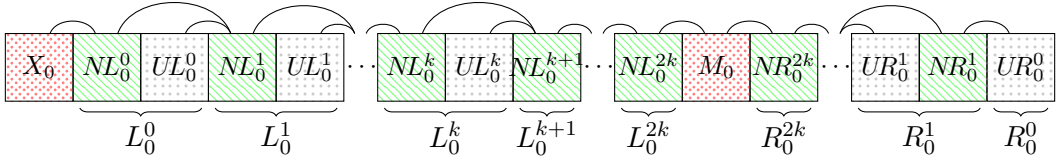

For every $\type\in\Freq$, let $\otypefirst,\ldots,\otypelast$ be an enumeration of \Ord.
Recall the definition of $M_k^d$ in Remark~\ref{rem:M}: we have $|\Ord|\leq M_k^d!$ for every $\type\in\Ntp$.

In particular, Proposition~\ref{prop:scatter} ensures that we are able to pick, for every $\type\in\Freq$, every $l\in\{1,\ldots,|\Ord|\}$ and every $j\in\{0,\ldots,k\}$, two elements $\pinLeft$ and $\pinRight$ which have $\type$ as $k$-neighbourhood type in \struct, such that all the $\pinLeft$ and the $\pinRight$, for $\type\in\Freq$, $l\in\{1,\ldots,|\Ord|\}$ and $j\in\{0,\ldots,k\}$, are at distance at least $4k+1$ from each other and from any occurrence of a rare $k$-neighbourhood type in \struct.

As their names suggest, $\pinLeft$ will appear in $\ULeft$ and $\pinRight$ in $\URight$, in order to make sure that each environment (i.e., ordering) of each frequent type appears in each universal segment $\ULeft, \URight$ -- which are called ``universal'' for that very reason.

\paragraph*{Construction of $\Rare$ and $\NLeft[0]$}
We start with the set $\Rare$, which contains all the occurrences of rare $k$-neighbourhood types, together with their $k$-neighbourhoods.
Formally, $\Rare$ is $\bigcup_{\elem\in\structdom :\ \ntp\in\Ntp\setminus\Freq}\boule\,.$
We set $\NLeft[0] \coloneqq \Neigh{\Rare}$ (the set of neighbours of elements of $\Rare$), and define the order \ordgen on $\Rare$ and on $\NLeft[0]$ in an arbitrary way.

\paragraph*{Construction of $\ULeft$}

If $k<j\leq 2k$, then we set $\ULeft \coloneqq \emptyset$.
Otherwise, for $j\in\{0,\ldots,k\}$, once we have constructed $\Left[0],\ldots,\Left[j-1]$ and $\NLeft$, we construct $\ULeft$ as follows.

The elements of $\ULeft$ are $\bigcup_{\tau\in\Freq}\bigcup_{l=1}^{|\Ord|}\boule[\pinLeft]\,.$
Note that $\ULeft$ does not intersect the previously constructed segments, by choice of the $\pinLeft$. Furthermore, the $\boule[\pinLeft]$ are pairwise disjoint, hence we can fix~\ordgen freely and independently on each of them.
Unsurprisingly, we order each $\boule[\pinLeft]$ so that $\envtp[\pinLeft]=\otype$. This is possible because for every $\type\in\Freq$ and each $l$, $\ntp[\pinLeft]=\type$ by choice of $\pinLeft$.

Once each $\boule[\pinLeft]$ is ordered according to $\otype$, the linear order \ordgen on $\ULeft$ can be completed in an arbitrary way.
Note that every possible $k$-environment type extending a frequent $k$-neighbourhood type in \struct occurs in each $\ULeft$. The $\ULeft$ are \emph{universal} in that sense.

\paragraph*{Construction of $\NLeft$}
Now, let us see how the $\NLeft$ are constructed. For $j\in\{1,\ldots,2k\}$, suppose that we have constructed $\Left[0],\ldots,\Left[j-1]$.
The domain of $\NLeft$ consists of all the neighbours (in \struct) of the elements of $\Left[j-1]$ not already belonging to the construction so far. Formally,
$\Neigh{\Left[j-1]}\setminus\big(\Rare\cup\bigcup_{m=0}^{j-2}\Left[m]\big).$

The order \ordgen on $\NLeft$ is chosen arbitrarily.

\paragraph*{Construction of $\Right$}

We construct similarly the $\Right$, for $j\in\{0,\ldots,2k\}$, starting with $\NRight[0] \coloneqq \emptyset$, then $\URight[0]$ which contains each $\pinRight[\otype][0]$ together with its $k$-neighbourhood in \struct ordered according to $\otype$, then $\NRight[1] \coloneqq \Neigh{\Right[0]}$, then $\URight[1]$, etc.

Note that the $\pinRight$ have been chosen so that they are far enough in \struct from all the segments that have been constructed so far, allowing us once more to order their $k$-neighbourhood in~\struct as we see fit.
  
\paragraph*{Construction of $\Middle$}

Let $\Middle$ contain all the elements of $\structdom$ besides those already belonging to $\Segment[2k]$.
The order \ordgen chosen on $\Middle$ is arbitrary.

\paragraph*{Transfer on \structbis}

Suppose that we have constructed $\Segment[2k]$.
We can make sure, retrospectively, that the index $f(k)$ in Equation (\ref{eq:foeq}) is large enough so that there exists a set $S\subseteq\structbisdom$ so that $\struct\mathord{\upharpoonright}_{\Segment[2k]\cup\Neigh{\Segment[2k]}}\simeq\structbis\mathord{\upharpoonright}_{S}$. This is ensured as long as $f(k)\geq|\Segment[2k]\cup\Neigh{\Segment[2k]}|+1$, whose right side can be bounded by a function of $k$, independent of \struct and \structbis. Indeed, all elements of $\Segment[2k]\cup\Neigh{\Segment[2k]}$ are at distance at most $3k+2$ from $\Rare$ or one of the $\pin$. Since the number of such elements is bounded by a function of $k$ (remember that, by Definition~\ref{def:frequent}, the number of occurrences of each rare neighbourhood types is at most $\thr(k)$) and since the degree is at most $d$ (which is fixed), $|\Segment[2k]\cup\Neigh{\Segment[2k]}|$ can be bounded by a function of $k$.

Let $\bijsymb:\struct\mathord{\upharpoonright}_{\Segment[2k]}\to\structbis\mathord{\upharpoonright}_{S'}$ be the restriction to $\Segment[2k]$ of said isomorphism, and let $\bijsymb[1]$ be its converse. By construction, the $k$-neighbourhood of every $\elem\in\Segment[k]$ is included in $\Segment[2k]$; hence every such $\elem$ has the same $k$-neighbourhood type in \struct as has $\bij[\elem]$ in \structbis.

We transfer alongside $\bijsymb$ all the segments besides $\Middle$, with their order, from \orderstruct to \structbis, thus defining $\Rarebis, \NLeftbis, \ULeftbis,\ldots$ as the respective images by $\bijsymb$ of $\Rare,\NLeft,\ULeft,\ldots$, and we define $\Middlebis$ as the set of remaining elements of $\structbisdom$. It remains to define the restriction of $\ordgen$ on $\Middlebis$: this is done in an arbitrary way. Note that the properties concerning neighbourhood are transferred; \eg all the neighbours of an element in $\Leftbis$, $1\leq j<2k$, belong to $\Leftbis[j-1]\cup\Leftbis\cup\Leftbis[j+1]\,.$

By construction, we get the following lemma:

\begin{lem}
  \label{lem:envpres}
  For each $\elem\in\Segment[k],$ we have $\envtp=\envtpbis[\bij[\elem]]\,.$
\end{lem}
Lemma~\ref{lem:envpres} has two immediate consequences:
\begin{itemize}
\item The set $\Rarebis$ contains the occurrences in \structbis of all the rare $k$-neighbourhood types (just forget about the order on the $k$-environments, and remember that \struct and \structbis have the same number of occurrences of each rare $k$-neighbourhood type).
\item All the universal segments $\ULeftbis$ and $\URightbis$, for $0\leq j\leq k$, contain at least one occurrence of each environment in \Ord, for each $\type\in\Freq$.
\end{itemize}

By virtue of $\bijsymb$ being an isomorphism, and because the order $\ordgen$ was constructed on $S'$ as the image by $\bijsymb$ of $\ordgen$, we get:

\begin{lem}
  \label{lem:tppres}
  For each $\elem,\elembis\in\Segment[k],$ we have $\tpgen[\elem,\elembis][\orderstruct]=\tpgen[\bij[\elem],\bij[\elembis]][\orderstructbis]\,.$
\end{lem}
In particular, for $\elem=\elembis\in\Segment[k]$, we have $\tp[\elem]=\tpbis[\bij[\elem]]\,.$

\subsection{\texorpdfstring{Proof of the \FOtwo-similarity of \orderstruct and \orderstructbis}{Proof of the FO2-similarity of (A0,<0) and (A1,<1)}}
\label{sec:EFproof}

In this section, we aim to show the following result:

\begin{prop}
  \label{prop:fotwoeq}
  $\orderstruct\fotwoeq\orderstructbis$
\end{prop}

Recall Theorem~\ref{th:EF}: in order to prove Proposition~\ref{prop:fotwoeq}, we show that the duplicator wins the $k$-round two-pebble \EF game on \orderstruct and \orderstructbis.
For that, let us show by a decreasing induction on $r$ from $k$ to $0$ that the duplicator can ensure, after $k-r$ rounds, that the three following properties (described below) hold. Recall that $\bijsymb$ and $\bijsymb[1]$ are the isomorphisms between $\struct\mathord{\upharpoonright}_{\structdom\setminus\Middle}$ and $\structbis\mathord{\upharpoonright}_{\structbisdom\setminus\Middlebis}$ along which the segments were transferred from \struct to \structbis.
\begin{align}
  &\forall\indin\forall\epsin\ \pebblegen\in\Segmentind\ \text{implies}\ \pebblegen[1-\ind]=\bijgen[\pebblegen] \tag{$\labsr$}\\
  &\forall\epsin\ \envtp[\pebblegen[0]][r]=\envtpbis[\pebblegen[1]][r] \tag{$\laber$}\\
  &\tp=\tpbis \tag{$\labtr$}
\end{align}
Let us stress that in this context, $\pebbleone$, $\pebbletwo$, $\pebbleonebis$ and $\pebbletwobis$ denote the elements of $\struct$ and $\structbis$ on which the pebbles have been played after $k-r$ rounds, and thus depend on $r$.

The first property, \refsr, guarantees that if a pebble is close (in a sense that depends on the number of rounds left in the game) to one of the \ordgen-minimal or \ordgen-maximal elements, the corresponding pebble in the other structure is located at the same position with respect to this \ordgen-extremal element.

As for \refer, it states that two corresponding pebbles are always placed on elements sharing the same $r$-environment type. Once again, the safety distance decreases with each round that goes.

Finally, \reftr controls that both pebbles have the same relative position (both with respect to the order and the original vocabulary) in the two ordered structures.
In particular, the duplicator wins the game if \reftr is satisfied at the beginning of the game, and after each of the $k$ rounds of the game.

\subsubsection{\texorpdfstring{Base case: proofs of \refsr[k], \refer[k] and \reftr[k]}{Base case: proofs of (Sk), (Ek) and (Rk)}}

We start by proving \refsr[k], \refer[k] and \reftr[k].
At the start of the game, the duplicator places both $\pebbleone$ and $\pebbletwo$ on the \ordgen-minimal element of \orderstruct, and both $\pebbleonebis$ and $\pebbletwobis$ on the \ordgen-minimal element of \orderstructbis. In particular,
\[\pebbleonebis=\pebbletwobis=\bij=\bij[\pebbletwo]\,.\]
This ensures that \refsr[k] holds, while \refer[k] and \reftr[k] follow from Lemmas~\ref{lem:envpres}--\ref{lem:tppres}.

\subsubsection{Winning strategy for the duplicator}
\label{subsec:strategy}

We now describe the duplicator's strategy to ensure that \refsr, \refer and \reftr hold no matter how the spoiler plays.

Suppose that we have \refsr[r+1], \refer[r+1] and \reftr[r+1] for some $0\leq r < k$, after $k-r-1$ rounds of the game.
Without loss of generality, we may assume that, in the $(k-r)$-th round of the \EF game between \orderstruct and \orderstructbis, the spoiler moves $\pebbleone$ in \orderstruct.

\begin{enumerate}[(I)]
\item\label{enu:border} If $\pebbleone\in\Segment$, then the duplicator responds by placing $\pebbleonebis$ on $\bij$.

  This corresponds to a tit-for-tat strategy. If the duplicator does not respond this way, then the spoiler will be able to expose the difference between \orderstruct and \orderstructbis in the subsequent moves, by forcing the duplicator to play closer and closer to the \ordgen-minimal or -maximal elements.

\item\label{enu:neigh} Else, if $\pebbleone\notin\Segment$, and $\pebbleone\in\voisdom$, then \refer[r+1] ensures that there exists an isomorphism $\psi:\Env[\pebbletwo][r+1]\to\Envbis[\pebbletwobis][r+1]\,.$ The duplicator responds by placing $\pebbleonebis$ on $\psi(\pebbleone)$ -- in other words, on the element whose relative position to~$\pebbletwobis$ is the same as the relative position of $\pebbleone$ with respect to $\pebbletwo$. 

\item\label{enu:farleft} Else suppose that $\orderstruct\models\pebbleone<\pebbletwo$ and $\pebbletwo\notin\Left[r+1]$. In other words, the spoiler places $\pebbleone$ to the left of $\pebbletwo$, when $\pebbletwo$ is not in the leftmost segment outside of $\Segment$ (namely, $\Left[r+1]$). The situation where $\pebbletwo\in\Left[r+1]$ is covered by Case~\ref{enu:left}.

  We have $\type \coloneqq \ntp[\pebbleone]\in\Freq$ as $\pebbleone\notin\Rare$. Let $\otype \coloneqq \envtp[\pebbleone]$.
  The duplicator responds by placing $\pebbleonebis$ on $\bij[\pinLeft[\otype][r+1]]$.

\item\label{enu:left} Else, if $\orderstruct\models\pebbleone<\pebbletwo$ and $\pebbletwo\in\Left[r+1]$, then the duplicator moves $\pebbleonebis$ on $\bij$ (by \refsr[r+1], $\pebbleone$ indeed belongs to the domain of $\bijsymb$).

\item\label{enu:farright} Else, suppose that $\orderstruct\models\pebbletwo<\pebbleone$ and $\pebbletwo\notin\Right[r+1]$. This case is symmetric to Case~\ref{enu:farleft}.

  Similarly, the duplicator opts to play $\pebbleonebis$ on $\bij[\pinRight[\otype][r+1]]$, where $\otype \coloneqq \envtp[\pebbleone]$.

\item\label{enu:right} If we are in none of the cases above, it means that the spoiler has placed $\pebbleone$ to the right of $\pebbletwo$, and that $\pebbletwo\in\Right[r+1]$. This case is symmetric to Case~\ref{enu:left}.

  Once again, the duplicator places $\pebbleonebis$ on $\bij$.
\end{enumerate}

It remains to show that this strategy satisfies our invariants: under the inductive assumption that \refsr[r+1], \refer[r+1] and \reftr[r+1] hold, for some $0\leq r<k$, we need to show that this strategy ensures that \refsr, \refer and \reftr hold.

We treat each case in its own section: Section~\ref{sec:border} is devoted to Case~\ref{enu:border} while Section~\ref{sec:neigh} covers Case~\ref{enu:neigh}. Both Cases~\ref{enu:farleft} and \ref{enu:left} are treated in Section~\ref{sec:left}. Cases~\ref{enu:farright} and \ref{enu:right}, being their exact symmetric counterparts, are left to the reader.

\begin{rem}
  \label{remark:y}
Note that some properties need no verification. Since $\pebbletwo$ and $\pebbletwobis$ are left untouched by the players, \refsr[r+1] ensures that half of \refsr automatically holds, namely that
\[\forall\indin,\quad\pebblegen[\ind][y]\in\Segmentind\quad\to\quad\pebblegen[1-\ind][y]=\bijgen[\pebblegen[\ind][y]]\,.\]
Similarly, the part of \refer concerning $\pebbletwo$ and $\pebbletwobis$ follows from \refer[r+1]:
\[\envtp[\pebbletwo][r]=\envtpbis[\pebbletwobis][r]\,.\]
Lastly, notice that once we have shown that \refer holds, it follows that $\vocabtpone=\vocabtponebis$ and $\vocabtptwo=\vocabtptwobis$. This remark will help us prove \reftr.
\end{rem}

\subsubsection{When the spoiler plays near the endpoints: Case~\ref{enu:border}}
\label{sec:border}

In this section, we treat the case where the spoiler places $\pebbleone$ near the \ordgen-minimal or \ordgen-maximal element of \orderstruct. Obviously, what ``near'' means depends on the number of rounds left in the game; the more rounds remain, the more the duplicator must be cautious regarding the possibility for the spoiler to reach an endpoint and potentially expose a difference between \orderstruct and~\orderstructbis.

As we have stated in Case~\ref{enu:border}, with $r$ rounds left, we consider a move on $\pebbleone$ by the spoiler to be near the endpoints if it is made in $\Segment$. In that case, the duplicator responds along the tit-for-tat strategy, namely by placing $\pebbleonebis$ on $\bij$.

Let us now prove that this strategy guarantees that \refsr, \refer and \reftr hold. Recall from Remark~\ref{remark:y} that part of the task is already taken care of.

\paragraph*{Proof of \refsr in Case~\ref{enu:border}}

We have to show that $\forall\indin,\ \pebblegen[\ind][x]\in\Segmentind\ \to\ \pebblegen[1-\ind][x]=\bijgen[\pebblegen[\ind][x]]\,.$
This follows directly from the duplicator's strategy, since $\pebbleonebis=\bij$ (thus $\pebbleone=\bijbis$).

\paragraph*{Proof of \refer in Case~\ref{enu:border}}

We need to prove that $\envtp[\pebbleone][r]=\envtpbis[\pebbleonebis][r]\,,$ which is a consequence of Lemma~\ref{lem:envpres} given that $\pebbleonebis=\bij$ and $r<k$.

\paragraph*{Proof of \reftr in Case~\ref{enu:border}}

First, suppose that $\pebbletwo\in\Segment[r+1]$. By \refsr[r+1], we know that $\pebbletwobis=\bij[\pebbletwo]$. Thus, Lemma~\ref{lem:tppres} allows us to conclude that $\tp=\tpbis$.

Otherwise, $\pebbletwo\notin\Segment[r+1]$ and \refsr[r+1] entails that $\pebbletwobis\notin\Segmentbis[r+1]$.
We have two points to establish:
\begin{align}
  \vocabtp=\vocabtpbis\label{eq:I_vocab},\\
  \ordtp=\ordtpbis\label{eq:I_order}.
\end{align}
By construction, the neighbours in $\structind$ of an element of $\Segmentind$ all belong to $\Segmentind[r+1]$, thus $\pebbleone$ and $\pebbletwo$ never appear in the same tuple of some relation, and neither do $\pebbleonebis$ and $\pebbletwobis$. 
In other words, the only atoms that hold in $\struct$ about $\pebbleone$ and $\pebbletwobis$ concern independently $\pebbleone$ and~$\pebbletwo$.
Similarly, the only atoms that hold in $\structbis$ about $\pebbleonebis$ and $\pebbletwobis$ concern independently $\pebbleonebis$ and $\pebbletwobis$. Equation~(\ref{eq:I_vocab}) follows from this fact together with $\vocabtpone=\vocabtponebis$ and $\vocabtptwo=\vocabtptwobis$ (Remark~\ref{remark:y}).

As for Equation~(\ref{eq:I_order}), either \[\pebbleone\in\Rare\cup\bigcup_{0\leq j\leq r}\Left\quad\text{and}\quad\pebbleonebis\in\Rarebis\cup\bigcup_{0\leq j\leq r}\Leftbis\,,\] in which case $\ordtp$ and $\ordtpbis$ both contain the atom ``$x<y$'' (as $\pebbleone$ and $\pebbleonebis$ appear in segments to the left of the segments containing $\pebbletwo$ and $\pebbletwobis$, respectively) and are thus equal, or \[\pebbleone\in\bigcup_{0\leq j\leq r}\Right\quad\text{and}\quad\pebbleonebis\in\bigcup_{0\leq j\leq r}\Rightbis\,,\] in which case $\ordtp$ and $\ordtpbis$ both contain the atom ``$x>y$'', and are still equal.

\subsubsection{When the spoiler plays next to the other pebble: Case~\ref{enu:neigh}}
\label{sec:neigh}

Suppose now that the spoiler places $\pebbleone$ next to the other pebble in \struct (\ie $\pebbleone\in\voisdom$), but not in $\Segment$ (for that move would fall under the jurisdiction of Case~\ref{enu:border}).
In that case, the duplicator must place~$\pebbleonebis$ so that the relative position of $\pebbleonebis$ and $\pebbletwobis$ is the same as that of $\pebbleone$ and $\pebbletwo$.

For that, we can use \refer[r+1], which guarantees that $\envtp[\pebbletwo][r+1]=\envtpbis[\pebbletwobis][r+1]\,.$ Thus there exists an isomorphism $\psi$ between $\Env[\pebbletwo][r+1]$ and $\Envbis[\pebbletwobis][r+1]$. Note that this isomorphism is unique, as we are dealing with linear orders.

The duplicator's response is to place $\pebbleonebis$ on $\psi(\pebbleone)$. Let us now prove that this strategy is correct with respect to our invariants \refsr, \refer and \reftr.

\paragraph*{Proof of \refsr in Case~\ref{enu:neigh}}

Because the spoiler's move does not fall under Case~\ref{enu:border}, we know that $\pebbleone\notin\Segment$.
Let us now show that $\pebbleonebis$ is not near the endpoints either, \ie that $\pebbleonebis\notin\Segmentbis$.

We consider two cases, depending on whether $\pebbletwobis\in\Segmentbis[r+1]$.
Since $\pebbleonebis$ and $\pebbletwobis$ are neighbours in \structbis, if $\pebbletwobis\notin\Segmentbis[r+1]$, then by construction $\pebbleonebis\notin\Segmentbis$.
Otherwise, $\pebbletwobis\in\Segmentbis[r+1]$ and we know by \refsr[r+1] that $\pebbletwo=\bijbis[\pebbletwobis]$. Because $\psi$ is the unique isomorphism between $\Env[\pebbletwo][r+1]$ and $\Envbis[\pebbletwobis][r+1]$, $\psi$ is equal to the restriction $\widetilde\bijsymb$ of $\bijsymb$:
\[\widetilde\bijsymb\ \colon \Env[\pebbletwo][r+1]\ \to\ \Envbis[\pebbletwobis][r+1]\,.\]
Thus $\pebbleone=\psi^{-1}(\pebbleonebis)=\widetilde\bijsymb^{-1}(\pebbleonebis)=\bijbis$, and by definition of the segments on \orderstructbis, which are just a transposition of the segments of \orderstruct via $\bijsymb$, $\pebbleone\notin\Segment$ then entails $\pebbleonebis\notin\Segmentbis$.

Since we neither have $\pebbleone\in\Segment$ nor $\pebbleonebis\in\Segmentbis$, \refsr holds -- recall from Remark~\ref{remark:y} that the part concerning $\pebbletwo$ and $\pebbletwobis$ is always satisfied.

\paragraph*{Proof of \refer in Case~\ref{enu:neigh}}

Recall that the duplicator placed $\pebbleonebis$ on the image of $\pebbleone$ by the isomorphism \[\psi:\Env[\pebbletwo][r+1]\to\Envbis[\pebbletwobis][r+1]\,.\]
It is easy to check that the restriction $\widetilde\psi$ of $\psi$:
$\widetilde\psi\ :\ \Env[\pebbleone][r]\ \to\ \Envbis[\pebbleonebis][r]$
is well-defined, and is indeed an isomorphism.\\ This ensures that $\envtp[\pebbleone][r]=\envtpbis[\pebbleonebis][r]\,,$ thus completing the proof of \refer.

\paragraph*{Proof of \reftr in Case~\ref{enu:neigh}}

This follows immediately from the fact that the isomorphism~$\psi$ maps $\pebbleone$ to $\pebbleonebis$ and $\pebbletwo$ to $\pebbletwobis$: all the atomic facts about these elements are preserved.

\subsubsection{When the spoiler plays to the left: Cases~\ref{enu:farleft} and \ref{enu:left}}
\label{sec:left}

We now treat our last case, which covers both Cases~\ref{enu:farleft} and \ref{enu:left}, \ie the instances where the spoiler places~$\pebbleone$ to the left of $\pebbletwo$ (formally: such that $\orderstruct\models\pebbleone<\pebbletwo$), which do not already fall in Cases~\ref{enu:border}--\ref{enu:neigh}.

Note that the scenario in which the spoiler plays to the right of the other pebble is the exact symmetric of this one (since the $\Rareind$ play no role in this case, left and right can be interchanged harmlessly).

The idea here is very simple: since the spoiler has placed $\pebbleone$ to the left of $\pebbletwo$, but neither in $\Segment$ nor in $\voisdom$, the duplicator responds by placing $\pebbleonebis$ on an element of $\ULeftbis[r+1]$ (the leftmost universal segment not in $\Segmentbis$) sharing the same $k$-environment.
This is possible by construction of the universal segments: if $\otype \coloneqq \envtp[\pebbleone]$ (which must extend a frequent $k$-neighbourhood type, since $\pebbleone\notin\Rare$), then $\bij[\pinLeft[\otype][r+1]]$ satisfies the requirements.

There is one caveat to this strategy. If $\pebbletwobis$ is itself in $\Leftbis[r+1]$, two problems may arise: first, it is possible for $\pebbleonebis$ and $\pebbletwobis$ to be in the wrong order (\ie such that $\orderstructbis\models\pebbleonebis>\pebbletwobis$). Second, it may be the case that $\pebbleonebis$ and $\pebbletwobis$ are neighbours in \structbis, which, together with the fact that $\pebbleone$ and $\pebbletwo$ are not neighbours in \struct, would break \reftr.

This is why the duplicator's strategy depends on whether $\pebbletwobis\in\Leftbis[r+1]$:
\begin{itemize}
\item if this is not the case, then the duplicator places $\pebbleonebis$ on $\bij[\pinLeft[\otype][r+1]]$. This corresponds to Case~\ref{enu:farleft}.
\item if $\pebbletwobis\in\Leftbis[r+1]$, then \refsr[r+1] guarantees that $\pebbletwo\in\Left[r+1]$. Hence $\pebbleone$, which is located to the left of $\pebbletwo$, is in the domain of $\bijsymb$: the duplicator moves $\pebbleonebis$ to $\bij$. This situation corresponds to Case~\ref{enu:left}.
\end{itemize}
Let us prove that \refsr, \refer and \reftr hold in both of these instances.

\paragraph*{Proof of \refsr in Case~\ref{enu:farleft}}

Since the spoiler's move does not fall under Case~\ref{enu:border}, we have that $\pebbleone\notin\Segment$.
By construction, $\pinLeft[\otype][r+1]\in\Left[r+1]$, thus $\bij[\pinLeft[\otype][r+1]]\in\Leftbis[r+1]$, and $\pebbleonebis\notin\Segmentbis$.

\paragraph*{Proof of \refer in Case~\ref{enu:farleft}}

It follows from $\envtp[\pinLeft[\otype][r+1]]=\otype$ together with Lemma~\ref{lem:envpres} that \[\envtp[\pebbleone]=\envtpbis[\pebbleonebis]\,.\]
A fortiori, $\envtp[\pebbleone][r]=\envtpbis[\pebbleonebis][r]$.

\paragraph*{Proof of \reftr in Case~\ref{enu:farleft}}

Because the spoiler's move does not fall under Case~\ref{enu:neigh}, $\pebbleone\notin\voisdom$. In other words, the (positive) atoms occurring in $\vocabtp$ are exactly those appearing in $\vocabtpone$ and $\vocabtptwo$, as no atom involving both $\pebbleone$ and $\pebbletwo$ holds.

Recall the construction of $\ULeft[r+1]$: the whole $k$-neighbourhood of $\pinLeft[\otype][r+1]$ was included in this segment. In particular, $\boule[\pebbleonebis][\structbis][1]=\boule[\bij[\pinLeft[\otype][r+1]]][\structbis][1]\subseteq\ULeftbis[r+1]\,.$
By assumption, $\pebbletwobis\notin\Leftbis[r+1]$, which entails that the atoms in $\vocabtpbis$ are those occurring in $\vocabtponebis$ and $\vocabtptwobis$.
It then follows from the last observation of Remark~\ref{remark:y} that $\vocabtp=\vocabtpbis\,.$

Let us now prove that $\ordtpbis$ contains the atom ``$x<y$''.
We claim that $\pebbletwobis\notin\Rarebis\cup\bigcup_{0\leq j\leq r+1}\Leftbis$. Suppose otherwise: \refsr[r+1] would entail that $\pebbletwo\in\Rare\cup\bigcup_{0\leq j\leq r+1}\Left$ which, together with the assumption $\pebbletwo\notin\Left[r+1]$ and $\pebbleone<\pebbletwo$, would result in $\pebbleone$ being in $\Segment$, which is a contradiction.
Both $\ordtpbis$ and $\ordtp$ contain the atom ``$x<y$'' and thus must be equal, which concludes the proof of \reftr.

\paragraph*{Proof of \refsr, \refer and \reftr in Case~\ref{enu:left}}

Let us now move to the case where $\pebbletwobis\in\Leftbis[r+1]$. Recall that under this assumption, $\pebbletwo=\bijbis[\pebbletwobis]\in\Left[r+1]$ and since $\pebbleone<\pebbletwo$ and $\pebbleone\notin\Segment$, we have that $\pebbleone\in\Left[r+1]$.
The duplicator places the pebble $\pebbleonebis$ on $\bij$; in particular, $\pebbleonebis\in\Leftbis[r+1]$.

The proof of \refsr follows from the simple observation that $\pebbleone\notin\Segment$ and $\pebbleonebis\notin\Segmentbis$.
As for \refer and \reftr, they follow readily from Lemmas~\ref{lem:envpres} and \ref{lem:tppres} and the fact that
$\pebbleonebis=\bij$ and $\pebbletwobis=\bij[\pebbletwo]$.

\subsection{Counting quantifiers}
\label{sec:counting}

We now consider the natural extension $\Ctwo$ of \FOtwo, where one is allowed to use counting quantifiers of the form $\existsm{i}x$ and $\existsm{i}y$, for $i\in\N$. Such a quantifier, as expected, expresses the existence of at least $i$ elements satisfying the formula which follows it.
This logic $\Ctwo$ has been extensively studied. From an expressiveness standpoint, $\Ctwo$ strictly extends \FOtwo (which cannot express the existence of three elements over the empty vocabulary), and contrary to the latter, \Ctwo does not enjoy the small model property (meaning that contrary to \FOtwo, there exist satisfiable $\Ctwo$-sentences which do not have small, or even finite, models). However, the satisfiability problem for $\Ctwo$ is still decidable~\cite{DBLP:conf/lics/GradelOR97,DBLP:journals/logcom/Pratt-Hartmann07,DBLP:conf/wollic/Pratt-Hartmann10}. To the best of our knowledge, it is not known whether $\oictwo$ has a decidable syntax.
Let us now explain how the proof of Theorem~\ref{th:main} can be adapted to show the following stronger version:

\begin{thm}
  \label{th:mainCtwo}
  Let \classe be a class of structures of bounded degree on a purely relational vocabulary.
  Then $\oictwo\subseteq\FO$ on \classe.
\end{thm}

\begin{proof}
  The proof is very similar to that of Theorem~\ref{th:main}. The difference is that we now need to consider \Ctwo-similarity instead of \FOtwo-similarity. To that end, we consider two families of equivalence relations, denoted $\ctwoeq$ and $\oictwoeq$ (for $k\in\N$), such that $\struct\ctwoeq\structbis$ and  $\struct\oictwoeq\structbis$ iff \struct and \structbis agree on all $\Ctwo$-sentences and \oictwo-sentences, respectively, of quantifier rank at most~$k$ and with counting indexes at most $k$. %Note that every $\oictwo$-sentence belongs to one of these classes of sentences.

  Let us rephrase Remark~\ref{rem:end_proof} in the case of \Ctwo. Assume we manage to show the existence of some function $\thr$ such that for $\struct,\structbis\in\classe$, one can construct $\ordgen$ with
  \begin{equation}
    \label{eq:ctwoeq}
    \orderstruct\ctwoeq\orderstructbis
  \end{equation}
  whenever $\threq{\struct}{\structbis}{k}{\thr(k)}$.  
  Once again considering $f:k\mapsto\hat f(k,\thr(k))$, Proposition~\ref{prop:FO_neigh} guarantees that $\struct\foeq[f(k)]\structbis$ implies $\struct\oictwoeq\structbis$.
  This amounts to saying that the equivalence relation $\foeq[f(k)]$ is more fine-grained than $\oictwoeq$ on \classe. An \oictwo-sentence of quantifier rank $k$ and with counting indexes at most $k$ can first be decomposed according to the equivalence classes of $\oictwoeq$, and then rewritten on \classe as a disjunction of \FO-sentences of rank $f(k)$ defining $\foeq[f(k)]$ classes, which concludes the proof.
  
  It only remains to show the existence of such a function $\thr$ and to detail the strategy for the construction of two linear orders satisfying Equation~(\ref{eq:ctwoeq}).

  In order to prove Equation~(\ref{eq:ctwoeq}), we need an \EF game capturing $\ctwoeq$. It is not hard to derive such a game from the \EF game for $\Ctwo$~\cite{immerman1990describing}.
  This game only differs from the two-pebble \EF game in that in each round, once the spoiler has chosen a structure (say \orderstruct) and a pebble to move (say $\pebbleone$), the spoiler picks not only one element of that structure, but a set $\set$ of up to $k$ elements. Then the duplicator must respond with a set $\setbis$ of same cardinality in \orderstructbis. The spoiler then places $\pebbleonebis$ on any element of $\setbis$, to which the duplicator responds by placing $\pebbleone$ on some element of $\set$. As usual, the spoiler wins after this round if $\tp\neq\tpbis\,.$ Otherwise, the game goes on until $k$ rounds are played.

  It is not hard to establish that this game indeed captures $\ctwoeq$, in the sense that $\orderstruct\ctwoeq\orderstructbis$ if and only if the duplicator has a winning strategy for $k$ rounds of this game. The restriction on the cardinality of the set chosen by the spoiler (which is at most~$k$) indeed corresponds to the fact that the counting indexes of the formulae are at most~$k$. As for the number of rounds (namely, $k$), it corresponds as usual to the quantifier rank. This can be easily derived from a proof of Theorem 5.3 in~\cite{immerman1990describing}, and is left to the~reader.

  Let us now explain how to modify the construction of \ordgen presented in Section~\ref{sec:construction} in order for the duplicator to maintain similarity for $k$-round in such a game. The only difference lies in the choice of the universal elements. Recall that in the previous construction, we chose, for each $k$-environment type $\tau_l$ extending a frequent $k$-neighbourhood type and each segment $\ULeft$, an element $\pinLeft$ whose $k$-environment type in \orderstruct is destined to be $\tau_l$ (and similarly for $\URight$ and $\pinRight$).

  In the new construction, we pick $k$ such elements, instead of just one. Just as previously, all these elements must be far enough from one another in the Gaifman graph of \struct. Once again, this condition can be met by virtue of the $k$-neighbourhood type $\tau$ underlying $\tau_l$ being frequent, and thus having many occurrences scattered across \struct (remember that we have a bound on the degree of \struct, thus all the occurrences of $\tau$ cannot be concentrated). We only need to multiply the value of $m$ by $k$ in Equation~(\ref{eq:m_and_k}). This in turn means that \thr must be larger than for the case of \FOtwo.
  
  When the spoiler picks a set of elements of size at most $k$ in one of the structures (say~$\set$ in \orderstruct), the duplicator responds by selecting, for each one of the elements of $\set$, an element in \orderstructbis along the strategy for the \FOtwo-game explained in Section~\ref{subsec:strategy}. All that remains to be shown is that it is possible for the duplicator to answer each element of $\set$ with a different element in \orderstructbis.

  Note that if the duplicator follows the strategy from Section~\ref{subsec:strategy}, they will never answer two moves by the spoiler falling under different cases among Cases~\ref{enu:border}-\ref{enu:right} with the same element. Thus we can treat separately each one of these cases; and for each case, we show that if the spoiler chooses up to $k$ elements in \orderstruct falling under this case in~$\set$, then the duplicator can find the same number of elements in \orderstructbis, following the aforementioned~strategy.
  \begin{itemize}
  \item For Case~\ref{enu:border}, this is straightforward, since the strategy is based on the isomorphism between the borders of the linear orders. The same goes for Cases~\ref{enu:neigh}, \ref{enu:left} and \ref{enu:right}, as the strategy in these cases also relies on an isomorphism argument.
  \item Suppose now that $\pebbletwo\notin\Left[r+1]$, and assume that the spoiler chooses several elements to the left of $\pebbletwo$, but outside of $\Segment$ and not adjacent to $\pebbletwo$. This corresponds to Case~\ref{enu:farleft}. Recall that our new construction guarantees, for each $k$-environment type extending a frequent $k$-neighbourhood type, the existence in $\Leftbis[r+1]$ of $k$ elements having this environment. This lets us choose, in $\Leftbis[r+1]$, a distinct answer for each element in the set selected by the spoiler, sharing the same $k$-environment type. Case~\ref{enu:farright} is obviously symmetric.
  \end{itemize}
  This concludes the proof of Theorem~\ref{th:mainCtwo}.
\end{proof}

\subsection{Constant symbols}
\label{sec:constants}
Theorem~\ref{th:main} and \ref{th:mainCtwo} assume the vocabulary is purely relational. We now explain how to lift this constraint.

To extend these theorems to a vocabulary \vocab containing constant symbols, it suffices to alter the definition of $k$-neighborhoods in the proof. Instead of including only elements at distance at most $k$ from \elem, $\neighgen$ now also includes the interpretations of all the (finitely many) constant symbols in \vocab. The only notable difference in the proof is an increase in the bound on the size of neighborhoods -- more precisely, $M_k^d$ has to be increased by the number of constant symbols in Remark~\ref{rem:M}. This increase has no bearing on the proof, as $M_k^d$ is still a function of \vocab, $k$ and $d$.

%!TEX root = ../order-invariant-FO2.tex

\section{Conclusion}\label{sec:conclusion}

In this paper, we made significant progress towards a better understanding of the two-variable fragment of order-invariant first-order logic:
\begin{itemize}\itemsep0em
\item From a complexity point of view, we established the $\coNET$-completeness of the problem of deciding if a given \FOtwo-sentence is order-invariant (Theorem~\ref{thm:complexity}), significantly simplifying and improving the result by Zeume and Harwath~\cite[Thm.~12]{ZeumeH16}.
\item From an expressivity point of view, we addressed the question of whether every property definable in order-invariant $\FOt$ can also be expressed in plain $\FO$.
We failed short of fully answering the question, but provided two interesting results.
The first one (namely, Theorem~\ref{thm:beyondfo}) establishes that under a more relaxed notion of order-invariance, the answer to the above question is ``no''. While this does not bring a fully-satisfactory answer to the problem, this leads us to believe that order-invariant $\FOt$ can indeed express properties beyond the scope of $\FO$.
The second one (Theorem~\ref{th:main}) states that when the degree is bounded, every property expressible in order-invariant $\FOt$ is definable in $\FO$ without the use of the order. This is an important step towards resolving the conjecture that order-invariant $\FO$ over classes of structures of bounded degree cannot express properties beyond the reach of $\FO$.
\end{itemize} 

Results of Section~\ref{sec:expr_power} also apply to the case of the two-variable logic with counting, $\Ctwo$. 
While order-invariant $\Ctwo$ has decidable satisfiability and validity problems~\cite[Theorem 6.20]{Charatonik016a}, it is open if it has a decidable syntax (\ie whether the problem of determining if a given $\Ctwo$-sentence is order-invariant is decidable). Unfortunately the techniques introduced in Section~\ref{sec:complexity} are of no use here, as $\Ctwo$ lacks the finite model property.
Finally, it might be a good idea to study order-invariant $\FOt$ over graph classes beyond classes of bounded-degree, \eg planar graphs or nowhere-dense classes of graphs.

\section*{Acknowledgements}
Bartosz Bednarczyk was supported by the ERC Consolidator Grant No.~771779 (DeciGUT).
He would like to thank Antti Kuusisto and Anna Karykowska for many insightful discussions on the problem.
The authors want to thank Thomas Schwentick and the anonymous reviewers for their helpful comments.

\bibliographystyle{alphaurl}
\bibliography{biblio}

\end{document}